\documentclass[11pt]{article}
\usepackage[margin=1in]{geometry}
\usepackage{amsmath,amssymb,amsthm,amsfonts,mathtools}
\usepackage{graphicx,booktabs,array}
\usepackage[numbers,sort&compress]{natbib}
\usepackage[dvipsnames]{xcolor}
\usepackage[colorlinks=true,allcolors={blue!55!black}]{hyperref}
\usepackage[capitalise,nameinlink]{cleveref}
\graphicspath{{figs/}}
\usepackage{float}       
\usepackage[section]{placeins} 

\theoremstyle{plain}
\newtheorem{theorem}{Theorem}[section]
\newtheorem{lemma}[theorem]{Lemma}
\newtheorem{proposition}[theorem]{Proposition}
\newtheorem{corollary}[theorem]{Corollary}
\theoremstyle{definition}
\newtheorem{definition}[theorem]{Definition}
\newtheorem{assumption}[theorem]{Assumption}
\theoremstyle{remark}
\newtheorem{remark}[theorem]{Remark}

\crefname{assumption}{Assumption}{Assumptions}
\crefname{remark}{Remark}{Remarks}
\numberwithin{equation}{section}

\DeclareMathOperator{\tr}{tr}

\DeclareMathOperator{\ran}{ran}
\DeclareMathOperator{\KL}{KL}
\DeclareMathOperator{\ESS}{ESS}

\newcommand{\R}{\mathbb{R}}
\newcommand{\E}{\mathbb{E}}
\newcommand{\Prob}{\mathbb{P}}
\newcommand{\dd}{\,\mathrm{d}}
\newcommand{\Res}{\mathcal R}
\newcommand{\Jc}{\mathcal J}          
\newcommand{\Jp}{\mathcal J^{+}}      
\newcommand{\Jmin}{\mathcal J^{-}}    
\newcommand{\Dt}{\widetilde D}
\newcommand{\Kdef}{\Upsilon}
\newcommand{\lamstar}{\lambda_\star}  
\newcommand{\Inc}{\mathsf B}
\newcommand{\Sg}{\mathsf S_g}
\newcommand{\Ip}{\mathfrak I}
\newcommand{\Dp}{\mathfrak D}

\newcounter{infobox}
\crefname{infobox}{Box}{Boxes}\Crefname{infobox}{Box}{Boxes}
\newenvironment{infobox}[1]{\par\medskip\noindent\refstepcounter{infobox}%
  \begin{center}\small\begin{tabular}{|p{0.94\linewidth}|}\hline
  \rule{0pt}{2.6ex}\textbf{Box~\theinfobox.\ #1}\par\smallskip}%
  {\\[2pt]\hline\end{tabular}\end{center}\medskip}

\title{How suboptimal is my stochastic network controller allowed to be?\\ {\large
Completion certificates with application to power grids hosting AI data
centers
}}
\author{Michael Chertkov\thanks{Program in Applied Mathematics and Department of
Mathematics, University of Arizona, Tucson, AZ 85721, USA
(\href{mailto:chertkov@arizona.edu}{chertkov@arizona.edu}; ORCID
0000-0002-6758-515X). This work was supported by the U.S. National Science
Foundation under award 2229012, \emph{Collaborative Research: AMPS: Rare Events
in Power Systems: Novel Mathematics, Statistics and Algorithms}.}}
\date{\today}

\begin{document}
\maketitle

\begin{abstract}
Power grids are beginning to host AI data centers whose demand can change abruptly and in a correlated way, and operators and planners must decide whether existing controllers can absorb the resulting transients and where new flexibility is worth installing. We cast this as a question in stochastic control: how far from optimal is an implementable controller? For controlled diffusions with affine, state-independent actuation, additive and possibly degenerate noise and quadratic control cost, any Hamilton–Jacobi–Bellman (HJB) subsolution bounds the optimal cost from below and simulation bounds the deployed cost from above, so their gap certifies the permissible suboptimality. We construct subsolutions from path-integral control by completing the control geometry: enlarging the control Gramian until it matches the physical noise makes the problem linearly solvable, and its Feynman–Kac value is an automatic lower bound whose HJB residual is exactly the energy of the fictitious control. A dual noise-deflation construction can be tighter but requires a curvature condition. The geometry yields planning rules: price control authority in proportion to local noise variance, and use shadow values to guide sparse reinforcement. For nonlinear stochastic swing dynamics of the IEEE 118-bus system after a severe load loss, a simple generator controller is certified within 2.8\% of optimal under homogeneous forcing; under heterogeneous forcing the gap is 39\% with uniform prices and 1.2\% once the same total authority is repriced—before any hardware is added.
\end{abstract}



\section*{Context and significance}
\label{sec:context}

\paragraph{A question at the meeting point of three fields}
This paper answers a control-theoretic question with tools from statistical
physics, and tests the answer on the problem of operating and planning power
grids. The three perspectives are not decoration; each supplies something the
others lack.

From \emph{control theory} comes the question and the notion of an answer. A
deployed controller is rarely optimal, and on a large network the optimum is
out of reach; what an engineer needs is a \emph{certificate}---a computable
bound on how much better any controller could possibly do. Such certificates
follow from subsolutions of the Hamilton--Jacobi--Bellman (HJB) equation, but
producing a good subsolution for a nonlinear, high-dimensional, underactuated
stochastic system is the whole difficulty.

From \emph{statistical physics} comes the construction. Path-integral control,
which originated in statistical physics \cite{kappen2005linear}, identifies
certain optimal-control problems with free energies of path ensembles: the value
function is a log-partition function, the control is an exponential tilt of the
uncontrolled paths, and regret is a relative entropy. These problems are exactly
solvable by sampling. The physical grid is not of that form, but it can be
\emph{completed} into one, and the completed problem yields a subsolution
automatically. The dictionary between the two languages is collected in
\Cref{box:physics}.

From \emph{power systems} comes the test and the stakes. The swing equations of a
transmission network are nonlinear, the noise enters at every bus while fast
active-power control sits at a few, and the structure of the network---its
Laplacian, its inertia, the placement of its actuators---makes the abstract
geometry of the construction concrete. It is also where the answer matters.

\paragraph{Why it matters now: grids that must tolerate a new kind of load}
Large computational loads, and above all data centres built for training and
serving artificial-intelligence models, are becoming a material part of
electricity demand \cite{smith2026datacenter}. They are unlike the loads for
which grids were designed. A training run synchronizes very many generators,
so the power drawn by a single site can change abruptly and in a correlated way,
and protective ride-through behaviour can remove a large block of load at once.
Reliability organizations have begun to report operational challenges of
exactly this kind, including unexpected large-load reductions with associated
frequency and voltage excursions \cite{nerc2026sra}. A grid that hosts such
loads must be operated---and, more importantly, \emph{planned}---to tolerate the
fluctuations and transients they create. The deterministic question of whether
the grid resynchronizes after a large load loss is treated in a companion paper
\cite{chertkov2026cyclespace}; the present paper addresses the stochastic
control question that follows it.

For \emph{operation}, the contribution is a way to decide whether the controllers
already in place are good enough. If the certified gap is a few per cent, no
redesign of the controller can recover more, and effort is better spent
elsewhere; if it is large, the certificate says how much is at stake. For
\emph{planning}, the contribution is larger. The same geometry that produces the
certificate identifies where fast flexibility---storage, inverter-based
resources, controllable demand, or the data centres themselves---has the
highest marginal value, and shows that how existing control authority is
\emph{priced} relative to the local fluctuations can matter as much as how much
authority is installed. On the IEEE 118-bus system after a severe load loss, the
same total generator authority gives a certified gap of $39\%$ when priced
uniformly and $1.2\%$ when priced in proportion to the noise. That is a planning
statement a system operator can act on before any new hardware is purchased.

\section{Introduction}
\label{sec:intro}

Consider the controlled diffusion
\begin{equation}
  \dd X_t=\bigl(f(X_t,t)+Bu_t\bigr)\dd t+\sigma\dd W_t,
  \qquad D:=\tfrac12\sigma\sigma^\top\succeq0,
  \label{eq:sde}
\end{equation}
with state $X_t\in\R^n$, control $u_t\in\R^m$, constant matrices $B$ and
$\sigma$, and cost
\begin{equation}
 J^u(x,t)=\E^u_{x,t}\!\left[
   \phi(X_T)+\int_t^T\Bigl(V(X_s,s)+\tfrac12u_s^\top Ru_s\Bigr)\dd s
 \right],\qquad R\succ0 .
 \label{eq:cost}
\end{equation}
The optimal value is $J^\star:=\inf_{u\in\mathcal U_{\rm ad}}J^u$ over the
admissible class $\mathcal U_{\rm ad}$ of \Cref{ass:standing}. Actuation is
affine and state independent, noise is additive, and the control penalty is
quadratic; the drift $f$ and the state costs $V$ and $\phi$ are otherwise
general, nonlinear and nonseparable. The diffusion $D$ may be singular. That is
the case of interest in network applications, where noise acts on some dynamical
components only and a smaller subset of those is actuated.

For large networks the optimal feedback is out of reach, and the controllers
actually deployed---droop rules, scaled linear-quadratic feedback, learned or
structurally restricted policies---are chosen for implementability. The
question this paper answers is therefore not how to compute the optimum, but:
\emph{how far from optimal is the controller we have?} The answer takes the
form of a certificate. Any smooth \emph{subsolution} $\Jc$ of the HJB equation
satisfies $\Jc\le J^\star$, while simulating the deployed policy $u$ gives
$J^u\ge J^\star$; hence
\[
  0\ \le\ J^u-J^\star\ \le\ J^u-\Jc .
\]
The right-hand side is computable and is an upper bound on the regret of the
deployed controller. The difficulty is entirely in producing a good subsolution
without solving the HJB equation. The contribution of the paper is a
construction that does this, together with its consequences for design.

\subsection{Contributions}

\paragraph{Linearly solvable completion (\Cref{sec:relaxations})}
Write $G:=BR^{-1}B^\top$ for the \emph{quadratic-control Gramian}, the matrix
multiplying the gradient-square term in the HJB Hamiltonian,\footnote{It should
not be confused with the finite-horizon controllability Gramian.} and
$Q:=2D=\sigma\sigma^\top$. Path-integral control (PIC)
\cite{kappen2005linear,todorov2006linearly,dvijotham2012unifying} is exactly
linearly solvable when $Q=\lambda G$ for a scalar temperature $\lambda>0$. Our
physical problems are not of this form, and typically not even close: noise may
act in directions no actuator reaches. \emph{Control completion} enlarges the
control Gramian from $G$ to $G_+:=Q/\lambda$, which is admissible whenever
$\lambda G\preceq Q$. The completed problem is exactly linearly solvable with
the physical noise, its value $\Jp_\lambda$ is a Feynman--Kac log-partition
function, and (\Cref{thm:completion})
\[
  \Res[\Jp_\lambda]\;=\;\tfrac12\bigl\|v^+_\lambda\bigr\|^2\;\ge\;0,
\]
where $\Res$ is the physical HJB residual and $v^+_\lambda$ is the optimal
\emph{fictitious} control of the completed problem. The completed value is
therefore a subsolution with no curvature, convexity or sign condition to
verify. The least inflation $\lambda=\lamstar$ gives the tightest bound in the
family, and the certificate gap splits exactly into the energy of the
fictitious control plus the feedback mismatch of the deployed policy
(\Cref{cor:gap-split}).

\paragraph{Noise deflation}
The dual construction keeps $G$ and shrinks $Q$ to $\lambda G$. Its value
$\Jmin_\lambda$ has physical residual $\tr[(D-\Dt)\nabla^2\Jmin_\lambda]$ with
$\Dt=\lambda G/2$ (\Cref{thm:deflation}): it is a subsolution exactly when the
value bends upward in the deleted noise directions. This can be tighter than
completion, but the curvature condition must be verified. The two
constructions are not ordered in general.

\paragraph{Design consequences (\Cref{sec:design})}
The same geometry turns certificates into planning rules. Among allocations of
a fixed total generator authority, pricing authority in proportion to local
noise variance maximizes the matching temperature and minimizes the completion
burden (\Cref{prop:noise-proportional}); a closed-form robust allocation handles
several noise scenarios (\Cref{prop:robust}); a fixed architecture is exactly
linearly solvable for two scenarios only if their covariances are proportional
(\Cref{prop:multiscenario}); and the completed problem supplies shadow prices
for sparse reinforcement (\Cref{prop:shadow}).

\paragraph{Computation and status (\Cref{sec:computing})}
We state exactly what is and is not certified numerically. The Monte Carlo
estimator of $\Jp$ is biased \emph{upward}---the unsafe direction for a lower
bound---by approximately $\lambda/(2\ESS)$ (\Cref{lem:jensen}); a
distribution-free finite-sample lower confidence bound exists
(\Cref{prop:lcb}) but is vacuous at the small matching temperatures met in
practice (\Cref{rem:lcb-vacuous}). \Cref{tab:status} lists the status of every
claim in the paper.

\paragraph{Application (\Cref{sec:swing,sec:numerics})}
For lossless stochastic swing dynamics with noise at every bus and active-power
control at generators only, the completion adds virtual power authority at load
buses, the matching temperature is $\lamstar=\min_g r_g\Gamma_g^2$, and inertia
cancels from the geometry. On the IEEE 39- and 118-bus systems after a severe
load-loss event we certify simple generator controllers and find that the
certificate gap is governed less by how much authority exists than by whether it
is priced in the same geometry as the noise.

\begin{figure}[t]
  \centering
  \includegraphics[width=0.78\textwidth]{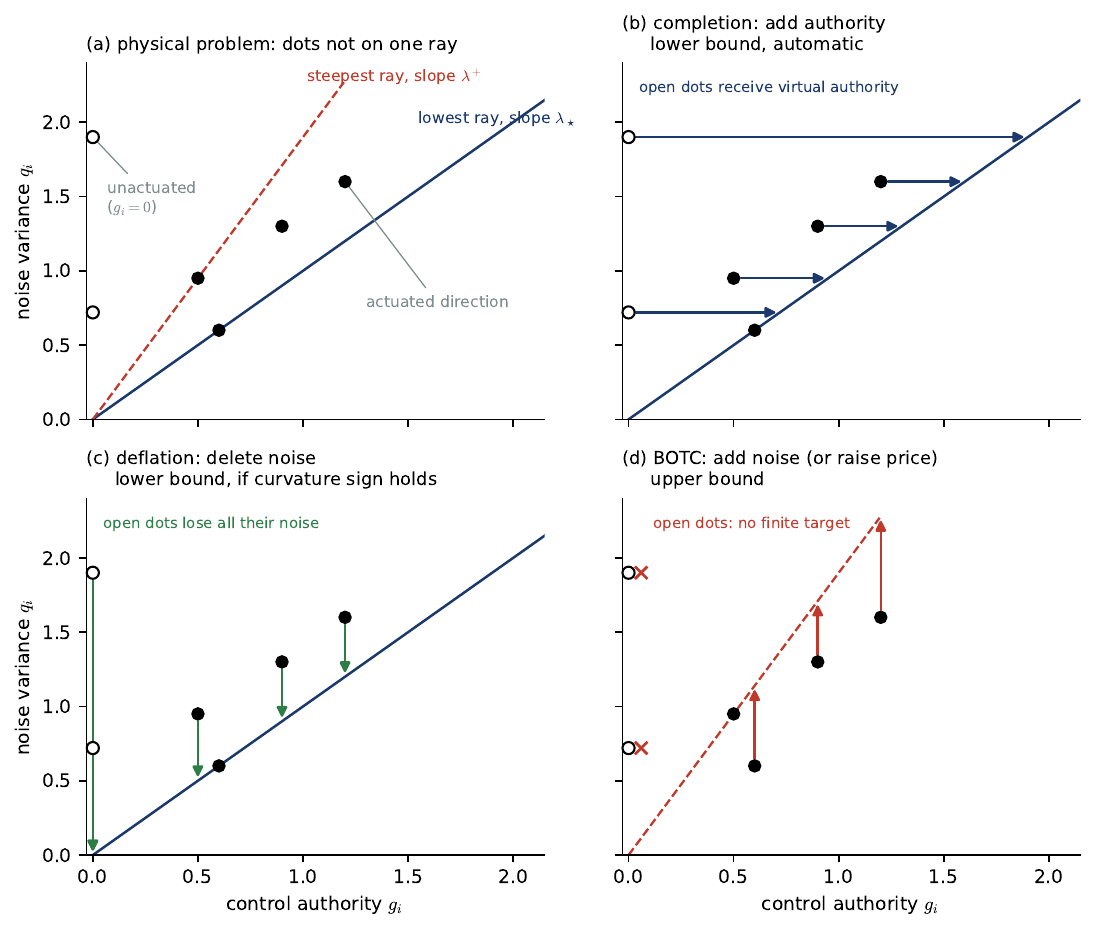}
  \caption{How the four constructions relate: a sketch, not a computation. Each dot
  		is one stochastic direction $i$ of a common eigenbasis of $Q$ and $G$ (one bus,
  		when noise and control channels are independent, as in \Cref{sec:swing}),
  		placed at its control authority $g_i$ (horizontal) and noise variance $q_i$
  		(vertical). Filled dots are \emph{actuated} directions ($g_i>0$: a controller can
  		push there); open dots are \emph{unactuated} directions ($g_i=0$: noise with no
  		actuator, such as the load buses of \Cref{sec:swing}). Dots are states of the
  		problem and arrows are the moves made by each construction. The problem is
  		exactly linearly solvable when all dots lie on one ray $q=\lambda g$ through the
  		origin. (a)~The physical problem is not: the lowest ray through an actuated dot
  		has slope $\lamstar$, the steepest has slope $\lambda^+$. (b)~\emph{Control
  			completion} (this paper) adds authority, moving every dot horizontally onto the
  		lowest ray; unactuated directions receive virtual authority. It gives an
  		automatic lower bound (\Cref{thm:completion}). (c)~\emph{Noise deflation} (this
  		paper) deletes noise, moving every dot vertically down onto the same ray;
  		unactuated directions lose all their noise. It gives a lower bound if the
  		curvature condition of \Cref{thm:deflation} holds. (d)~Bound-Optimized Task
  		Choice \cite{anderson2026botc} makes the problem harder, adding noise to move
  		the dots up onto the steepest ray (raising control prices, not shown, moves them
  		left instead), and gives upper bounds. Unactuated directions have no finite steepest ray, so BOTC has no move for them;
  		the red crosses in~(d) mark these directions, which constitute the support
  		mismatch discussed in \Cref{sec:botc}.}
  \label{fig:geometry}
\end{figure}

\subsection{Relation to existing work}

Verification theorems for HJB equations, and the use of sub- and supersolutions
to bound value functions, are classical
\cite{fleming2006controlled,crandall1992users,bayraktar2013perron}.
Finite-dimensional versions search polynomial or moment relaxations
\cite{parrilo2000structured,lasserre2008nonlinear,jones2020polynomial}, and the
linear structure of linearly solvable control has been combined with
sum-of-squares programming to produce guaranteed value bounds and
suboptimality guarantees \cite{horowitz2014semidefinite,leong2016linearly}. Our
certificates are of the same logical type. The difference is the source of the
bounding function: rather than approximating the original HJB equation, we
change the control--noise geometry until the HJB equation becomes exactly
linearly solvable, and evaluate that surrogate by Feynman--Kac. The residual in
the physical HJB is then explicit---fictitious control energy for completion,
deleted-noise curvature for deflation---and it identifies concrete directions,
prices and locations. The closest recent work is Bound-Optimized Task Choice
(BOTC) \cite{anderson2026botc}, which approaches the same matching obstruction
from the opposite side, by \emph{hardening} the problem to obtain upper bounds; \Cref{fig:geometry} sketches all four constructions in a single picture, in
which each stochastic direction is a point and each construction is a move.

On the application side, we use the lossless swing model of synchronization
theory \cite{dorfler2013synchronization,dorfler2018electrical} on standard test
systems \cite{zimmerman2011matpower}. The deterministic post-event equilibrium
and its cohesive geometry are taken from a companion study of load-loss
transients \cite{chertkov2026cyclespace}; large computational loads, whose
abrupt reductions motivate the event studied here, are an increasingly material
part of system operation \cite{smith2026datacenter,nerc2026sra}.

\subsection{How to read the paper}

\Cref{sec:certificates,sec:relaxations,sec:computing} contain the control
theory and are self-contained. \Cref{sec:design} translates it into design
rules, which are summarized for engineering readers in \Cref{box:engineering}.
\Cref{sec:swing,sec:numerics} are the power-system application; the modelling
scope is stated once, in \Cref{sec:scope}. Readers from statistical physics may
prefer to start from the dictionary in \Cref{box:physics}. The status of every
claim---theorem, exact identity, Monte Carlo estimate of an exact quantity, or
empirical diagnostic---is collected in \Cref{tab:status}.
\Cref{sec:app-deflation} records how far the curvature condition for the
deflated value has been verified in the application. Twelve further figures, each with a description of what it shows and what we conclude from it, are collected in \cref{app:figures}.

\paragraph{Notation}
$\succeq$ is the positive-semidefinite (PSD) order; $A^\dagger$ is the
Moore--Penrose pseudoinverse and $A^{\dagger/2}$ the pseudoinverse square root
of a PSD matrix $A$; $\lambda_{\max}(\cdot)$ is the largest eigenvalue;
$\|v\|_R^2=v^\top Rv$. Calligraphic $\Jc$ denotes a \emph{candidate}
cost-to-go, as distinct from the cost $J^u$ of a policy.

\section{Performance identities and subsolution certificates}
\label{sec:certificates}

\begin{assumption}[Standing assumptions]
\label{ass:standing}
$f$ is continuous and locally Lipschitz in $x$, $V$ and $\phi$ are continuous
and bounded below, and $B$, $\sigma$, $R\succ0$ are constant. The admissible
class $\mathcal U_{\rm ad}$ consists of progressively measurable controls for
which \eqref{eq:sde} has a unique strong solution, the cost \eqref{eq:cost} is
finite, and the integrability required by Dynkin's formula holds for the
candidate functions considered below.\footnote{For $F\in C^{2,1}$ and a stopping
time $\tau$, Dynkin's formula reads
$\E^u_{x,t}[F(X_\tau,\tau)]-F(x,t)=\E^u_{x,t}\int_t^\tau(\partial_s+\mathcal
L^u)F(X_s,s)\dd s$ with $\mathcal L^uF=(f+Bu)\cdot\nabla F+\tr(D\nabla^2F)$,
whenever the terms are integrable \cite{karatzas1991brownian,fleming2006controlled}.}
\end{assumption}

For a candidate cost-to-go $\Jc\in C^{2,1}(\R^n\times[0,T))$ define the physical
HJB residual and the feedback generated by the candidate,
\begin{align}
 \Res[\Jc]
 &:=\partial_t\Jc+f\cdot\nabla\Jc+\tr(D\nabla^2\Jc)+V
   -\tfrac12\nabla\Jc^\top G\,\nabla\Jc,
 \qquad G:=BR^{-1}B^\top,
 \label{eq:residual}\\
 u_{\Jc}(x,t)&:=-R^{-1}B^\top\nabla\Jc(x,t).
 \label{eq:uJ}
\end{align}
The HJB equation for \eqref{eq:sde}--\eqref{eq:cost} is $\Res[J^\star]=0$ with
$J^\star(\cdot,T)=\phi$.

\begin{proposition}[Performance identity]
\label{prop:identity}
Under \cref{ass:standing}, for every $u\in\mathcal U_{\rm ad}$,
\begin{equation}
 J^u(x,t)=\Jc(x,t)
 +\E^u_{x,t}\!\bigl[\phi(X_T)-\Jc(X_T,T)\bigr]
 +\E^u_{x,t}\!\int_t^T\!\Bigl[
   \Res[\Jc](X_s,s)+\tfrac12\|u_s-u_{\Jc}(X_s,s)\|_R^2
 \Bigr]\dd s .
 \label{eq:identity}
\end{equation}
\end{proposition}

\begin{proof}
Dynkin's formula applied to $\Jc(X_s,s)$ gives
$\E^u[\Jc(X_T,T)]-\Jc(x,t)=\E^u\int_t^T(\partial_s\Jc+(f+Bu)\cdot\nabla\Jc
+\tr(D\nabla^2\Jc))\dd s$. Insert this into \eqref{eq:cost} and complete the
square,
$u^\top B^\top\nabla\Jc+\tfrac12u^\top Ru
=\tfrac12\|u-u_{\Jc}\|_R^2-\tfrac12\nabla\Jc^\top G\nabla\Jc$.
\end{proof}

\Cref{prop:identity} is a form of the classical verification identity
\cite{fleming2006controlled}. Verification arguments set the terminal and HJB
defects to zero and recover the optimal feedback; here the defects are kept and
measured, and the identity is applied to candidates that are \emph{not}
solutions.

\begin{definition}[Subsolution]
\label{def:subsolution}
A candidate $\Jc\in C^{2,1}$ satisfying the integrability of
\cref{ass:standing} is a (classical) \emph{subsolution} if
\begin{equation}
  \Res[\Jc]\ge0\ \text{ on }\R^n\times[0,T),
  \qquad \Jc(\cdot,T)\le\phi .
  \label{eq:subsolution}
\end{equation}
\end{definition}

The sign convention is chosen so that a subsolution can only understate the
cost. Viscosity sub- and supersolutions
\cite{crandall1992users,fleming2006controlled} extend the notion when classical
regularity is absent; the constructions below produce classical candidates and
do not require them.

\begin{theorem}[Two-sided certificate]
\label{thm:certificate}
If $\Jc$ is a subsolution, then for every $u\in\mathcal U_{\rm ad}$
\begin{equation}
  \Jc(x,t)\ \le\ J^\star(x,t)\ \le\ J^u(x,t),
  \qquad\text{hence}\qquad
  0\ \le\ J^u-J^\star\ \le\ J^u-\Jc .
  \label{eq:bracket}
\end{equation}
\end{theorem}

\begin{proof}
All three correction terms in \eqref{eq:identity} are nonnegative, so
$J^u\ge\Jc$ for every admissible $u$; take the infimum over $u$.
\end{proof}

We call $J^u-\Jc$ the \emph{certificate gap} of the pair $(u,\Jc)$. It is an
upper bound on the regret $J^u-J^\star$ of the deployed controller. When the
optimum is itself classical, \eqref{eq:identity} with $\Jc=J^\star$ prices the
regret exactly.

\begin{corollary}[Regret is quadratic in feedback error]
\label{cor:regret}
If $J^\star\in C^{2,1}$ solves the HJB equation and
$u^\star=-R^{-1}B^\top\nabla J^\star$, then
\begin{equation}
 J^u(x,t)-J^\star(x,t)
 =\tfrac12\,\E^u_{x,t}\!\int_t^T\|u_s-u^\star(X_s,s)\|_R^2\dd s .
 \label{eq:regret}
\end{equation}
\end{corollary}

The expectation in \eqref{eq:regret} is taken under the \emph{deployed} closed
loop, so the identity is not a fixed Hilbert-space norm shared by all policies.
Its operational content is that a small feedback error has no first-order
effect on performance; \Cref{sec:implementation} uses it to price
implementation defects.

\begin{lemma}[Residual shift]
\label{lem:shift}
If $\Res[\Jc]\ge-\delta$ on $\R^n\times[0,T)$ for a constant $\delta\ge0$ and
$\Jc(\cdot,T)\le\phi$, then $\Jc_\delta(x,t):=\Jc(x,t)-\delta(T-t)$ is a
subsolution.
\end{lemma}

\begin{proof}
$\Res[\Jc_\delta]=\Res[\Jc]+\delta\ge0$, and $\Jc_\delta(\cdot,T)=\Jc(\cdot,T)$.
\end{proof}

A verification that returns a bounded negative residual is therefore not a
failure: it yields a certificate degraded by exactly $\delta(T-t)$.

\begin{remark}[The set of subsolutions is convex]
\label{rem:convex}
$\Jc\mapsto\Res[\Jc]$ is concave, since its only nonlinearity is
$-\tfrac12\nabla\Jc^\top G\nabla\Jc$; hence the set of classical subsolutions is
convex. It is not a cone: the running cost $V$ is fixed while the gradient term
scales quadratically. Convexity is a property of the \emph{set of functions}; it
should not be confused with convexity of a given $\Jc$ in the state, which
enters only in \Cref{thm:deflation}.
\end{remark}

\paragraph{Stopped domains}
For a domain $\Omega\subset\R^n$ and $\tau=\inf\{s\ge t:X_s\notin\Omega\}\wedge T$,
Dynkin's formula up to $\tau$ gives the exit-time analogue of \eqref{eq:identity}
\cite{fleming1977exit,fleming2006controlled}, with $\phi$ replaced by a
prescribed cost $g_\partial$ on the parabolic boundary. Thus $\Res[\Jc]\ge0$ in
$\Omega$ and $\Jc\le g_\partial$ on the boundary again yield a lower
certificate. We use this only for the separate first-exit question of
\Cref{sec:app-firstexit}; the performance value used throughout is smooth and
unstopped, for reasons explained in \Cref{sec:power-objective}.

\paragraph{Path-space form}
Let $u,v$ be admissible controls from the same initial state with path laws
$\Prob^u,\Prob^v$. If the drift change lies in the diffusion range,
$B(u_s-v_s)\in\ran(\sigma)$ almost surely, the minimum-norm Brownian shift is
$a_s=\sigma^\dagger B(u_s-v_s)$, and under the usual conditions of Girsanov's
theorem \cite{karatzas1991brownian}
\begin{equation}
 \KL(\Prob^u\Vert\Prob^v)=\tfrac12\E^u\!\int_t^T
 \bigl\|\sigma^\dagger B(u_s-v_s)\bigr\|^2\dd s .
 \label{eq:girsanov}
\end{equation}
If the range condition fails, the drift change is not an absolutely continuous
change of measure. When the problem is exactly linearly solvable and
nonredundant, $\|\sigma^\dagger B\,\delta u\|^2=\lambda^{-1}\|\delta u\|_R^2$, and
\eqref{eq:regret}--\eqref{eq:girsanov} identify regret with $\lambda$ times a
path-space relative entropy \cite{boue1998variational}. \Cref{box:physics}
collects this and the other statistical-mechanics readings.

\section{Two linearly solvable relaxations}
\label{sec:relaxations}

\Cref{thm:certificate} says what a certificate is, but not how to obtain one for a
large nonlinear problem. Path-integral control supplies two constructions,
dual to each other. Write
\begin{equation}
 Q:=2D=\sigma\sigma^\top,\qquad G:=BR^{-1}B^\top .
 \label{eq:QG}
\end{equation}
The problem \eqref{eq:sde}--\eqref{eq:cost} is exactly linearly solvable with
scalar temperature $\lambda$ when $Q=\lambda G$
\cite{kappen2005linear,todorov2006linearly,dvijotham2012unifying}: the Hopf--Cole
substitution $\Jc=-\lambda\log\psi$ then removes the gradient-square term of the
HJB equation and leaves a linear backward equation for $\psi$. \emph{Matched}
here means equality of the two state-space quadratic forms (hence of their
ranges), not equality of the numbers of control and noise channels. The
physical pair $(Q,G)$ is generally not matched. Both relaxations below rest on
the same inequality $\lambda G\preceq Q$, read in the two possible directions.

\begin{lemma}[Admissible matching temperatures]
\label{lem:lambdastar}
There is $\lambda>0$ with $\lambda G\preceq Q$ if and only if
$\ran(G)\subseteq\ran(Q)$. In that case the set of such $\lambda$ is
$(0,\lamstar]$, with
\begin{equation}
 \lamstar=\Bigl[\lambda_{\max}\bigl(Q^{\dagger/2}\,G\,Q^{\dagger/2}\bigr)\Bigr]^{-1}.
 \label{eq:lambdastar}
\end{equation}
\end{lemma}

\begin{proof}
If $w\in\ker Q$ has $w^\top Gw>0$, then $\lambda w^\top Gw\le w^\top Qw=0$ fails for
every $\lambda>0$; so $\ran(G)\subseteq\ran(Q)$ is necessary. Conversely, on
$\ran(Q)$ the inequality $\lambda G\preceq Q$ is equivalent to
$\lambda\,Q^{\dagger/2}GQ^{\dagger/2}\preceq P_Q$, the orthogonal projector onto
$\ran(Q)$, i.e.\ to $\lambda\,\lambda_{\max}(Q^{\dagger/2}GQ^{\dagger/2})\le1$.
\end{proof}

The condition $\ran(G)\subseteq\ran(Q)$ says that every direction the controller
can push is a direction in which the plant is stirred. It is much weaker than
nondegeneracy of $\sigma$, and it holds in the swing application of
\Cref{sec:swing}, where $\sigma$ is singular.

\subsection{Control completion: an automatic lower bound}
\label{sec:completion}

Fix $\lambda\in(0,\lamstar]$ and define the completed Gramian and its increment
\begin{equation}
  G_+(\lambda):=\frac{Q}{\lambda},\qquad
  \Delta G_\lambda:=G_+(\lambda)-G\ \succeq\ 0 .
  \label{eq:Gplus}
\end{equation}
Choose any factor $C_\lambda$ with $C_\lambda C_\lambda^\top=\Delta G_\lambda$ and
augment the plant by a \emph{fictitious} control $v$,
\begin{equation}
 \dd X_t=\bigl(f(X_t,t)+Bu_t+C_\lambda v_t\bigr)\dd t+\sigma\dd W_t,
 \qquad\text{running cost }\ V+\tfrac12u^\top Ru+\tfrac12\|v\|^2 ,
 \label{eq:completed-sde}
\end{equation}
with the same terminal cost $\phi$. ``Completion'' refers to the effective
quadratic-control geometry in state space; it does not require the number of
hardware channels to equal the number of Brownian coordinates.

\begin{assumption}[Completed problem]
\label{ass:completed}
For the given $\lambda$, the Feynman--Kac functional
\begin{equation}
 \psi_\lambda(x,t):=\E^{0}_{x,t}\exp\!\Bigl[-\frac1\lambda\Bigl(
 \phi(X_T)+\int_t^T V(X_s,s)\dd s\Bigr)\Bigr],
 \label{eq:FK-completed}
\end{equation}
taken over the \emph{uncontrolled} physical diffusion
$\dd X=f\dd t+\sigma\dd W$, belongs to $C^{2,1}(\R^n\times[0,T))$ and is a
classical solution of
\begin{equation}
 \partial_t\psi+f\cdot\nabla\psi+\tfrac12\tr(Q\nabla^2\psi)-\lambda^{-1}V\psi=0,
 \qquad \psi(\cdot,T)=e^{-\phi/\lambda},
 \label{eq:linear-pde}
\end{equation}
and $\Jp_\lambda:=-\lambda\log\psi_\lambda$ satisfies the integrability of
\cref{ass:standing}.
\end{assumption}

Since $V$ and $\phi$ are bounded below, $0<\psi_\lambda\le\exp(-S_{\min}/\lambda)$
with $S_{\min}=\inf\phi+T\inf V$, so finiteness and positivity are automatic. \Cref{rem:hypoelliptic} gives a sufficient condition for the
regularity.

\begin{theorem}[Completion certificate]
\label{thm:completion}
Let $\lambda\in(0,\lamstar]$ and suppose \cref{ass:completed} holds. Then
\begin{enumerate}
\item the physical HJB residual of $\Jp_\lambda$ is
\begin{equation}
 \Res[\Jp_\lambda]
 =\tfrac12\,\nabla\Jp_\lambda{}^{\!\top}\Delta G_\lambda\,\nabla\Jp_\lambda
 =\tfrac12\,\|v^+_\lambda\|^2\ \ge\ 0,
 \qquad v^+_\lambda:=-C_\lambda^\top\nabla\Jp_\lambda ,
 \label{eq:completion-residual}
\end{equation}
and $\Jp_\lambda(\cdot,T)=\phi$; hence
\item $\Jp_\lambda$ is a subsolution of the physical problem and
$\Jp_\lambda\le J^\star$, with no curvature, convexity or sign condition to
verify;
\item if in addition the feedback $(u^+_\lambda,v^+_\lambda)$, with
$u^+_\lambda:=-R^{-1}B^\top\nabla\Jp_\lambda$, is admissible for
\eqref{eq:completed-sde}, then $\Jp_\lambda$ is the value of the completed
problem and $v^+_\lambda$ is its optimal fictitious control.
\end{enumerate}
\end{theorem}

\begin{proof}
Set $\Jc=-\lambda\log\psi_\lambda$. A direct computation using
\eqref{eq:linear-pde} gives
$\partial_t\Jc+f\cdot\nabla\Jc+\tfrac12\tr(Q\nabla^2\Jc)+V
-\tfrac1{2\lambda}\nabla\Jc^\top Q\nabla\Jc=0$, i.e.\ the HJB equation with
Gramian $Q/\lambda=G+\Delta G_\lambda$. Since $\tfrac12\tr(Q\nabla^2\Jc)
=\tr(D\nabla^2\Jc)$, comparison with \eqref{eq:residual} gives
$\Res[\Jc]=\tfrac12\nabla\Jc^\top\Delta G_\lambda\nabla\Jc
=\tfrac12\|C_\lambda^\top\nabla\Jc\|^2$, which is item~1; the terminal
condition is read off from \eqref{eq:linear-pde}. Item~2 is then
\Cref{thm:certificate}. For item~3, the equation just derived is the HJB
equation of \eqref{eq:completed-sde}, and \Cref{prop:identity} applied to the
completed problem shows that every admissible completed control costs at least
$\Jp_\lambda$, with equality for the admissible feedback
$(u^+_\lambda,v^+_\lambda)$.
\end{proof}

Only items~1--2 are used for certification, and they need no optimality
argument for the completed problem: the certificate is a consequence of the
linear equation \eqref{eq:linear-pde} alone.

The performance identity now takes a particularly transparent form, in terms
of the physical component $u^+_\lambda=-R^{-1}B^\top\nabla\Jp_\lambda$ and the
fictitious component $v^+_\lambda$ of the completed feedback.

\begin{corollary}[The certificate gap splits exactly]
\label{cor:gap-split}
For every deployed $u\in\mathcal U_{\rm ad}$,
\begin{equation}
 J^u-\Jp_\lambda
 =\underbrace{\tfrac12\,\E^u\!\int_t^T\|v^+_\lambda(X_s,s)\|^2\dd s}_{\text{structural}}
 +\underbrace{\tfrac12\,\E^u\!\int_t^T\|u_s-u^+_\lambda(X_s,s)\|_R^2\dd s}_{\text{policy}} .
 \label{eq:gap-split}
\end{equation}
\end{corollary}

\begin{proof}
Insert \eqref{eq:completion-residual} and $\Jp_\lambda(\cdot,T)=\phi$ into
\eqref{eq:identity}.
\end{proof}

The first term is the price of the fictitious authority the completion had to
grant; no physical controller can remove it. The second is attributable to the
deployed policy. The fictitious control is therefore not merely a relaxation
variable: its energy along the deployed closed loop is exactly the part of the
certificate gap that reflects the mismatch between actuation and noise.

\begin{proposition}[Least inflation is tightest]
\label{prop:least-inflation}
For $0<\lambda_1<\lambda_2\le\lamstar$, $G_+(\lambda_1)\succeq G_+(\lambda_2)\succeq
G$, and
\begin{equation}
 \Jp_{\lambda_1}\ \le\ \Jp_{\lambda_2}\ \le\ J^\star .
 \label{eq:least-inflation}
\end{equation}
Hence $\lambda=\lamstar$ gives the least inflated Gramian and the tightest lower
bound within the scalar completion family.
\end{proposition}

\begin{proof}
A larger Gramian is a larger control set at no greater price: factor
$G_+(\lambda_1)-G_+(\lambda_2)=CC^\top$ and add a channel $w$ priced
$\tfrac12\|w\|^2$; setting $w\equiv0$ recovers the $\lambda_2$ problem.
\end{proof}

We write $\Jp:=\Jp_{\lamstar}$, with architecture labels when several actuator
classes are compared.

\begin{definition}[Completion diagnostics]
\label{def:diagnostics}
At $\lambda=\lamstar$, the \emph{completion burden} and the \emph{whitened
noise-direction deficit} are
\begin{equation}
 \Ip:=\tr\bigl(\Delta G_{\lamstar}\bigr)=\tr\!\Bigl(\frac{Q}{\lamstar}-G\Bigr),
 \qquad
 \Dp:=\tr\!\bigl(P_Q-\lamstar\,Q^{\dagger/2}GQ^{\dagger/2}\bigr).
 \label{eq:diagnostics}
\end{equation}
\end{definition}

The matrix in $\Dp$ is PSD with eigenvalues in $[0,1]$, so $\Dp$ counts, in
noise-whitened coordinates, stochastic directions that are uncovered or only
partly covered by actuation; it is invariant under $Q\mapsto c^2Q$. $\Ip$ keeps
physical units and also records how much the single temperature $\lamstar$
inflates the directions that \emph{are} actuated. Neither is a bound on
$J^\star-\Jp$; they are geometry diagnostics, and \Cref{sec:numerics} shows that
$\Ip$ is the informative one in practice.

\begin{remark}[Regularity of the completed value]
\label{rem:hypoelliptic}
The regularity part of \cref{ass:completed} is a statement about the linear
equation for $\psi_\lambda$ in the proof of \Cref{thm:completion}. If $f$, $V$,
$\phi$ are smooth and the vector fields $\sigma e_1,\dots,\sigma e_k$, together
with their iterated brackets with $\partial_t+f\cdot\nabla$, span $\R^n$ at every
point (the parabolic H\"ormander condition), that equation is hypoelliptic
\cite{hormander1967hypoelliptic} and $\psi_\lambda$ is smooth for $t<T$; since
$\psi_\lambda>0$, so is $\Jp_\lambda$. This covers degenerate noise that acts
only on velocities of a mechanical system, as in \Cref{sec:swing}. What remains
to be assumed is the global integrability required by Dynkin's formula.
\end{remark}

\begin{remark}[Path-space reading]
\label{rem:entropy-oracle}
Because $G_+=Q/\lambda$, the completed problem can be written as a drift
control in the Wiener directions, $\dd X=(f+\sigma a)\dd t+\sigma\dd W$ with cost
$\tfrac\lambda2\int\|a\|^2$, which by \eqref{eq:girsanov} is $\lambda$ times the
relative entropy of the controlled path law. The inequality $\lambda G\preceq Q$
is exactly what embeds the physical controls in this geometry without
increasing their price. The completed problem is thus the entropy-regularized
oracle obtained by granting control in every stochastic direction.
\end{remark}

\subsection{Noise deflation: a conditional lower bound}
\label{sec:deflation}

The dual move keeps $G$ and replaces $D$ by the matched diffusion
$\Dt(\lambda):=\tfrac\lambda2 G$, with deleted diffusion
\begin{equation}
  \Kdef(\lambda):=D-\Dt(\lambda)\ \succeq\ 0\qquad(0<\lambda\le\lamstar).
  \label{eq:deleted}
\end{equation}
The deflated problem is exactly linearly solvable; under the analogue of
\cref{ass:completed} for the \emph{deflated} diffusion its value is
$\Jmin_\lambda=-\lambda\log\widetilde\psi_\lambda$, with $\widetilde\psi_\lambda$
the Feynman--Kac functional \eqref{eq:FK-completed} computed under
$\dd X=f\dd t+\widetilde\sigma\dd W$, $\widetilde\sigma\widetilde\sigma^\top=2\Dt$.

\begin{theorem}[Deflation residual]
\label{thm:deflation}
The physical HJB residual of the deflated value is
\begin{equation}
  \Res[\Jmin_\lambda]=\tr\bigl(\Kdef(\lambda)\,\nabla^2\Jmin_\lambda\bigr),
  \qquad \Jmin_\lambda(\cdot,T)=\phi .
  \label{eq:deflation-residual}
\end{equation}
Consequently $\Jmin_\lambda$ is a physical subsolution if and only if
$\tr(\Kdef\nabla^2\Jmin_\lambda)\ge0$ on $\R^n\times[0,T)$.
\end{theorem}

\begin{proof}
Subtract the deflated HJB equation from \eqref{eq:residual}: drift, state costs
and the gradient-square term cancel, and only $\tr((D-\Dt)\nabla^2\Jmin_\lambda)$
remains.
\end{proof}

The condition is \emph{directional}: only curvature in $\ran(\Kdef)$, the
deleted noise directions, matters, and global convexity of $\Jmin_\lambda$ is
sufficient but unnecessary. If verification yields only
$\tr(\Kdef\nabla^2\Jmin_\lambda)\ge-\delta$, \Cref{lem:shift} turns
$\Jmin_\lambda-\delta(T-t)$ into a certificate. For affine drift the condition is
automatic.

\begin{proposition}[Affine drift and convex costs]
\label{prop:prekopa}
Suppose $f$ is affine in $x$, $V(\cdot,t)$ and $\phi$ are convex, and the
Feynman--Kac functional is finite. Then $\widetilde\psi_\lambda$ is log-concave
in $x$, $\Jmin_\lambda$ is convex, and $\Jmin_\lambda$ is a subsolution.
\end{proposition}

\begin{proof}[Proof sketch]
With affine drift the time-discretized path is affine jointly in the initial
state and the Gaussian increments, so the discretized action is jointly convex
and its exponential, times the Gaussian density, jointly log-concave.
Pr\'ekopa's theorem \cite{prekopa1973logarithmic} makes the marginal in the
initial state log-concave; log-concavity passes to the limit of the
discretization when the functional is finite. Convexity of $\Jmin_\lambda$ and
$\Kdef\succeq0$ give $\tr(\Kdef\nabla^2\Jmin_\lambda)\ge0$.
\end{proof}

For nonlinear drift, convexity of the costs is not enough, and the curvature
condition becomes part of the certificate.

\begin{remark}[The two relaxations are not ordered]
\label{rem:not-ordered}
\Cref{prop:least-inflation} has no deflation analogue: $\lamstar$ keeps the most
diffusion in the PSD order, but changing the diffusion changes the path law and
the nonlinear value need not be monotone. $\Jp$ and $\Jmin$ are therefore not
ordered in general, and \Cref{sec:numerics} exhibits both orders. Deflation
leaves the physical actuators untouched and can be tighter; completion is
unconditionally valid. In the larger of our two test systems completion is
also the tighter of the two, by a factor between three and ten, which is why
it carries the certificates reported there.
\end{remark}

\begin{infobox}{Statistical-mechanics dictionary}
\label{box:physics}
For readers from statistical physics, the objects of \Cref{sec:relaxations}
translate as follows. The matching temperature $\lambda$ is a temperature; the
Feynman--Kac functional $\psi_\lambda$ is a path-space partition function, and
the value $\Jp_\lambda=-\lambda\log\psi_\lambda$ is the corresponding free energy.
An optimal control is an exponential tilt of the uncontrolled path measure, and
by \eqref{eq:girsanov} regret in a matched problem is $\lambda$ times a relative
entropy. The value Hessian is a covariance under the tilted path measure,
$\partial_{ij}^2\Jp=\E^{\rm tilt}[\partial^2_{ij}S]-\lambda^{-1}{\rm
Cov}^{\rm tilt}(\partial_iS,\partial_jS)$ for the action $S$, which is why
curvature conditions such as \Cref{thm:deflation} are statements about
fluctuations. Control completion grants control in every stochastic direction:
it replaces the physical problem by the one in which the controller may act
wherever the noise does, at the entropic price.
\end{infobox}

\section{From certificates to design}
\label{sec:design}

For an existing system the fictitious directions of \Cref{thm:completion} are
virtual: they enlarge the controller class only to produce a lower bound. The
same PSD increment has a second reading in planning. If fast bidirectional
actuators can be installed---storage, inverter-based resources, controllable
demand---part or all of the completion can be built. This section develops that
reading. It does not claim that exact linear solvability is itself an
engineering objective; it asks what the completion geometry says about where
authority should sit and how it should be priced.

Three different notions of price appear in this section, and only one of them is
economic before any monetization: the control price $R$ is a quadratic weight on
control effort; the price of linearisability $\Pi$ is a gap between a restricted
and an unrestricted optimum, in the tradition of the price of anarchy, of
robustness and of fairness
\cite{koutsoupias1999worst,roughgarden2002bad,bertsimas2004price,bertsimas2011price};
and the shadow price of \Cref{prop:shadow} is a marginal value of control
authority, the only one with a direct economic reading.

Throughout, write $J^\star_G$ for the optimal value when only the control
Gramian is varied, with $f,D,V,\phi$ fixed.

\begin{proposition}[More authority never hurts]
\label{prop:monotone}
If $G_2\succeq G_1$, then $J^\star_{G_2}\le J^\star_{G_1}$.
\end{proposition}

\begin{proof}
Factor $G_2-G_1=CC^\top$ and add a channel $w$ priced $\tfrac12\|w\|^2$; setting
$w\equiv0$ recovers the $G_1$ controller class.
\end{proof}

\subsection{The price of linearisability}

Let $G_0$ be the existing Gramian, $H(y)\succeq0$ the increment delivered by a
design $y$, and $C_{\rm cap}(y)$ its cost. At budget $b$ define the unrestricted
planning value and its linearly solvable restriction,
\begin{equation}
 \mathcal V(b):=\inf_{C_{\rm cap}(y)\le b}J^\star_{G_0+H(y)},
 \qquad
 \mathcal V_{\rm PIC}(b):=\inf_{\substack{C_{\rm cap}(y)\le b,\ \lambda>0\\
 G_0+H(y)=Q/\lambda}}\Jp_\lambda ,
 \label{eq:planning}
\end{equation}
and, when both are finite, the \emph{price of linearisability}
$\Pi(b):=\mathcal V_{\rm PIC}(b)-\mathcal V(b)\ge0$. A small $\Pi$ would justify
designing for exact solvability on physical as well as computational grounds; a
large $\Pi$ means the completion should remain a certificate and a design
oracle. Within the linearly solvable family the tradeoff is one-dimensional:
\begin{equation}
 \lambda\longmapsto\Bigl(C_{\rm cap}\bigl(Q/\lambda-G_0\bigr),\ \Jp_\lambda\Bigr),
 \qquad 0<\lambda\le\lamstar(G_0),
 \label{eq:frontier}
\end{equation}
is an exactly computable investment--performance frontier: decreasing $\lambda$
installs more authority in the PSD order and, by \Cref{prop:least-inflation},
lowers $\Jp_\lambda$. We do not compute $\Pi$ in this paper, because
$\mathcal V(b)$ requires solving the unrestricted planning problem.

\subsection{Shadow prices of authority}

\begin{proposition}[Marginal value of authority]
\label{prop:shadow}
Suppose $J^\star_G$ is classical, the optimal closed loop is well posed, and
$\epsilon\mapsto J^\star_{G+\epsilon H}$ is differentiable at $\epsilon=0$ for a
direction $H\succeq0$, with enough domination to differentiate the HJB equation
and apply Feynman--Kac. Then
\begin{equation}
 \frac{\dd}{\dd\epsilon}J^\star_{G+\epsilon H}(x,t)\Big|_{\epsilon=0}
 =-\tfrac12\,\E^\star_{x,t}\!\int_t^T
 \nabla J^{\star\top}_GH\,\nabla J^\star_G\dd s .
 \label{eq:shadow}
\end{equation}
\end{proposition}

\begin{proof}
Differentiating the HJB equation in $\epsilon$ gives a linear backward equation
for the derivative under the optimal closed-loop generator, with source
$-\tfrac12\nabla J^{\star\top}_GH\nabla J^\star_G$; Feynman--Kac gives
\eqref{eq:shadow}.
\end{proof}

For a single candidate device, $H=aa^\top$ and the integrand is
$\tfrac12(a^\top\nabla J^\star_G)^2$. At a completed design the value and its
gradient are available from the linear representation, so the completed
problem supplies shadow prices that rank sparse reinforcements without
requiring the full completion to be built. \Cref{sec:numerics} compares such a
ranking with ranking by local noise.

\subsection{Pricing existing authority}
\label{sec:allocation}

The single temperature $\lamstar$ is set by the worst-matched actuated
direction, so how existing authority is priced matters as much as how much of
it there is. Consider independent channels on a set $\mathcal C$ of controlled
coordinates, with $Q=\operatorname{diag}(q_j)$ over all stochastic coordinates
and $G=\operatorname{diag}(g_i)$ supported on $\mathcal C$, where $g_i=1/r_i$ and
$q_i>0$ for $i\in\mathcal C$. Then \eqref{eq:lambdastar} reduces to
\begin{equation}
 \lamstar=\min_{i\in\mathcal C}\frac{q_i}{g_i},
 \qquad
 \Ip=\frac{\sum_jq_j}{\lamstar}-\sum_{i\in\mathcal C}g_i .
 \label{eq:diag-lambda}
\end{equation}

\begin{proposition}[Noise-proportional pricing]
\label{prop:noise-proportional}
Among allocations with fixed total authority $\sum_{i\in\mathcal C}g_i=A$, both
$\lamstar$ is maximized and $\Ip$ is minimized by
\begin{equation}
 g_i^\star=A\,\frac{q_i}{\sum_{j\in\mathcal C}q_j},
 \qquad
 \lamstar=\frac{\sum_{j\in\mathcal C}q_j}{A},
 \qquad
 \Ip=A\,\frac{\sum_{j\notin\mathcal C}q_j}{\sum_{j\in\mathcal C}q_j}.
 \label{eq:noise-proportional}
\end{equation}
Equivalently, $r_iq_i$ is equalized across controlled channels.
\end{proposition}

\begin{proof}
$\lambda\le q_i/g_i$ for all $i\in\mathcal C$ is $g_i\le q_i/\lambda$; summing,
$\lambda\le\sum_{\mathcal C}q_j/A$, with equality at $g^\star$. Since $\Ip$ in
\eqref{eq:diag-lambda} is decreasing in $\lamstar$ at fixed $A$, it is minimized
at the same allocation, where it takes the stated value.
\end{proof}

The last formula is worth reading literally: after optimal pricing, the
completion burden is the total authority times the ratio of \emph{uncontrolled}
to \emph{controlled} noise variance. What remains is genuine underactuation;
everything removed by the reallocation was mispricing.

For several noise scenarios $q_{s,j}$ with $S_s:=\sum_jq_{s,j}$, set
$c_i:=\min_s q_{s,i}/S_s$ for $i\in\mathcal C$.

\begin{proposition}[Robust fixed pricing]
\label{prop:robust}
Among allocations with $\sum_{i\in\mathcal C}g_i=A$, the worst-case completion
burden $\max_s\Ip_s$ is minimized by $g_i^{\rm rob}=A\,c_i/\sum_{j\in\mathcal
C}c_j$, with minimax value $A/\sum_{j\in\mathcal C}c_j-A$.
\end{proposition}

\begin{proof}
By \eqref{eq:diag-lambda}, $\Ip_s=S_s\max_{i\in\mathcal C}(g_i/q_{s,i})-A$. With an
epigraph variable $t$, the constraints $S_sg_i/q_{s,i}\le t$ for all $s,i$ read
$g_i\le tc_i$; the least feasible $t$ under $\sum g_i=A$ is $t=A/\sum c_j$,
attained by $g^{\rm rob}$, and at that allocation $\max_s\Ip_s=t-A$.
\end{proof}

\begin{proposition}[One architecture, several scenarios]
	\label{prop:multiscenario}
	If a fixed $G\neq0$ is exactly matched to two covariances on a common support,
	$Q_s=\lambda_sG$ for $s=1,2$, then $Q_2=(\lambda_2/\lambda_1)Q_1$.
\end{proposition}

\begin{proof}
	By definition a matching temperature is positive, so $\lambda_1>0$ and the first
	matching relation can be solved for the Gramian: $G=\lambda_1^{-1}Q_1$.
	Substituting this into the second relation, $Q_2=\lambda_2G$, gives
	$Q_2=(\lambda_2/\lambda_1)\,Q_1$.
	
	The same conclusion can be stated in terms of the \emph{shape} of the noise,
	that is, the covariance normalized by its trace. Because $G\succeq0$ and
	$G\neq0$, $\tr G>0$; taking traces in $Q_s=\lambda_sG$ gives
	$\tr Q_s=\lambda_s\tr G>0$, and hence
	\[
	\frac{Q_s}{\tr Q_s}=\frac{G}{\tr G},\qquad s=1,2 .
	\]
	Exact matching therefore fixes the noise shape to be the shape of $G$, leaving
	only the overall amplitude, carried by $\lambda_s$, free. Two scenarios can both
	be exactly matched to the same $G$ only if their shapes coincide, which is the
	statement of the proposition. (The common-support hypothesis is in fact
	automatic: $Q_s=\lambda_sG$ with $\lambda_s>0$ forces $\ran Q_s=\ran G$.)
\end{proof}

In the diagonal setting of \Cref{sec:allocation}, with $q_{s,i}=\lambda_sg_i$,
the proposition says that the ratio $q_{2,i}/q_{1,i}=\lambda_2/\lambda_1$ must be
the same at every controlled coordinate: scenarios whose noise redistributes
across buses, rather than merely rescaling, cannot share one exactly matched
pricing. At least one of them then requires completion, which is why
\Cref{prop:robust} optimizes a compromise instead.

\begin{remark}[Repricing changes the problem]
\label{rem:repricing}
\Cref{prop:noise-proportional,prop:robust} change $R$, and $R$ enters the
physical cost \eqref{eq:cost}. Comparing certificates before and after
repricing therefore compares two different control problems. What the
propositions license is a design statement: \emph{if} the operator is free to
set the prices of existing authority, setting them in proportion to the noise
yields a problem whose simple controllers can be certified tightly. They do
\emph{not} by themselves establish that a loose certificate under the original
prices was loose because of mispricing rather than because the original
controller is far from optimal; under the original prices $J^\star$ is
unknown. \Cref{sec:numerics} reports both readings where the data allow.
\end{remark}

\subsection{Implementation defects are priced in the same currency}
\label{sec:implementation}

\Cref{cor:regret} prices deviations of an implemented feedback from a reference
without solving a new control problem, and the same logic applies to the gap
\eqref{eq:gap-split} with $u^+_\lambda$ as reference. For example, for linear
dynamics with optimal linear feedback $u^\star=-K^\star(t)x$ implemented by
zero-order hold with interval $\Delta$, a first-order expansion of
\eqref{eq:regret} gives
\begin{equation}
 J^{\rm hold}-J^\star
 =\frac\Delta4\int_0^T\tr\bigl(K^{\star\top}RK^\star\sigma\sigma^\top\bigr)\dd t
 +O(\Delta^2),
 \label{eq:sample-hold}
\end{equation}
where $\tfrac14$ combines the mean-square Brownian increment over a hold
interval with the factor $\tfrac12$ of \eqref{eq:regret}. Lag, quantisation,
saturation, dropout and event-triggered sampling
\cite{tabuada2007event,heemels2012introduction} are priced analogously, and
\Cref{sec:numerics} measures them directly.

\begin{infobox}{Engineering summary}
\label{box:engineering}
The design content of \Cref{sec:relaxations,sec:design}, with the numerical
evidence of \Cref{sec:numerics}:\\
$\bullet$ E1: \emph{Certify before tuning.} The certificate gap bounds how much any
  controller, however designed, could improve on the deployed one. On the IEEE
  118-bus study it is $2.8\%$ for homogeneous forcing: further tuning cannot
  gain more (\Cref{sec:numerics}).\\
$\bullet$ E2: \emph{Price authority in proportion to local noise variance}
  (\Cref{prop:noise-proportional}). At fixed total generator authority this
  took the certified gap under heterogeneous forcing from $39\%$ to $1.2\%$;
  read it as a design choice, not a diagnosis (\Cref{rem:repricing}).\\
  $\bullet$  E3: \emph{One fixed pricing across scenarios} (\Cref{prop:robust}) gave
  $5.3\%$ under heterogeneous and $4.8\%$ under homogeneous forcing.\\
$\bullet$ E4: \emph{Check the pricing before buying hardware.} The single
  temperature is set by the worst-matched channel; adding actuators with weak
  local noise can make the scalar completion looser, not tighter.\\
$\bullet$ E5: \emph{Rank sparse devices by shadow value} (\Cref{prop:shadow})
  rather than by local noise when the forcing is heterogeneous. It is a
  heuristic, not a guarantee: in the 118-bus study the gain is statistically
  resolved but below $0.1\%$ of the cost, and under homogeneous forcing an
  arbitrary ordering did slightly better.\\
$\bullet$ E6: \emph{Implementation defects cost second order}
  (\Cref{sec:implementation}): a $0.2$\,s hold or lag cost about $2\%$ of the
  objective in the 39-bus study.\\
$\bullet$ E7:  \emph{Keep performance and safety separate.} The certificates are for
  a smooth performance objective; first-crossing probabilities are reported
  empirically and are not certified.
\end{infobox}

\section{Computing a certificate, and what is certified}
\label{sec:computing}

A certificate has three ingredients with different status. The \emph{upper}
value $J^u$ is an expectation under the deployed closed loop and is estimated by
plain Monte Carlo. The \emph{lower} value $\Jp$ is exact mathematically
(\Cref{thm:completion}) but is itself an expectation \eqref{eq:FK-completed},
estimated numerically. For the deflated value $\Jmin$ there is in addition a
\emph{residual verification} problem. We treat the three in turn.

\subsection{Estimating the completed value}

Passive sampling of \eqref{eq:FK-completed} suffers weight degeneracy when
$\lambda$ is small, and the natural remedy is importance sampling with a
controlled proposal, whose likelihood ratio is supplied by Girsanov's theorem
\cite{kappen2005linear,boue1998variational,dupuis2004importance}. For degenerate
diffusions the proposal drift must stay in $\ran(\sigma)$. In all cases the
estimator has the form
\begin{equation}
 \widehat{\Jp}_N=-\lambda\log\overline Y_N,\qquad
 \overline Y_N=\frac1N\sum_{k=1}^NY_k,\qquad \E\,Y_k=\psi_\lambda ,
 \label{eq:estimator}
\end{equation}
with $Y_k\ge0$ the weighted path functionals. Its bias has a sign.

\begin{lemma}[The estimator overstates the lower bound]
\label{lem:jensen}
If $Y_1,\dots,Y_N$ are independent copies of $Y\ge0$ with $\E Y=\psi_\lambda$ and
$\E|\log\overline Y_N|<\infty$, then $\E\,\widehat{\Jp}_N\ge\Jp_\lambda$. If
moreover $\E Y^3<\infty$, then as $N\to\infty$
\begin{equation}
 \E\,\widehat{\Jp}_N-\Jp_\lambda
 =\frac{\lambda}{2N}\,\frac{\operatorname{Var}Y}{\psi_\lambda^2}+o(N^{-1})
 \;\approx\;\frac\lambda2\Bigl(\frac1{\ESS}-\frac1N\Bigr),
 \qquad
 \ESS:=\frac{(\sum_kY_k)^2}{\sum_kY_k^2}.
 \label{eq:jensen-bias}
\end{equation}
\end{lemma}

\begin{proof}
The first claim is Jensen's inequality for the concave logarithm. The expansion
is the second-order delta method for $\log\overline Y_N$ about $\log\psi_\lambda$;
the last form replaces $\E Y^2/\psi_\lambda^2$ by its natural estimate
$N/\ESS$.
\end{proof}

The bias is in the unsafe direction for a lower bound, so it must be reported,
and \eqref{eq:jensen-bias} shows it is small whenever $\lambda/\ESS$ is small
compared with the certificate gap. The estimate is, however, blind to
probability mass carried by weights that were never sampled, which is the real
danger when $\ESS$ is small. A rigorous finite-sample statement requires control
of the weights.

\begin{proposition}[A finite-sample lower confidence bound]
\label{prop:lcb}
Suppose $0\le Y_k\le b$ almost surely. Then for every $\delta\in(0,1)$, with
probability at least $1-\delta$,
\begin{equation}
 J^\star\ \ge\ \Jp_\lambda\ \ge\
 \widehat{\Jp}_{\rm LCB}:=-\lambda\log\Bigl(\overline Y_N
 +b\sqrt{\tfrac{\log(1/\delta)}{2N}}\Bigr).
 \label{eq:lcb}
\end{equation}
Boundedness holds for passive sampling with $b=e^{-S_{\min}/\lambda}$, where
$S_{\min}=\inf\phi+T\inf V$, and for a defensive mixture proposal
$\alpha\Prob^0+(1-\alpha)\Prob^{\rm g}$, whose likelihood ratio is at most
$1/\alpha$ \cite{hesterberg1995weighted}, with $b=e^{-S_{\min}/\lambda}/\alpha$.
\end{proposition}

\begin{proof}
Hoeffding's inequality gives
$\Prob\{\psi_\lambda>\overline Y_N+b\sqrt{\log(1/\delta)/(2N)}\}\le\delta$; apply
the decreasing map $y\mapsto-\lambda\log y$ and \Cref{thm:completion}.
\end{proof}

\begin{remark}[The distribution-free bound is vacuous at low temperature]
\label{rem:lcb-vacuous}
The Hoeffding term in \eqref{eq:lcb} scales with the weight \emph{range} $b$,
while $\psi_\lambda\approx e^{-\Jp_\lambda/\lambda}$. Even in the most favourable
case $S_{\min}=0$ the ratio $b/\psi_\lambda\approx e^{\Jp_\lambda/\lambda}$ is
astronomically large at the matching temperatures met in the application. For
the hardest case of \Cref{sec:numerics} ($\lamstar\approx1.76\times10^{-4}$,
$\Jp\approx0.025$) it is about $e^{142}$, and \eqref{eq:lcb} reduces to
$\widehat{\Jp}_{\rm LCB}\approx\tfrac\lambda2\log\bigl(2N/\log(1/\delta)\bigr)$:
with $N=2048$ and $\delta=0.05$ this is about $6\times10^{-4}$, a valid but
useless bound. Sharper empirical-Bernstein variants do not change the picture,
because the difficulty is the range, not the variance. A useful rigorous bound
at low temperature needs a pathwise lower bound on the action that is close to
its typical value, which we do not have. We therefore report $\Jp$ as a Monte
Carlo estimate of an exact lower value, with replicate standard errors, the
effective sample size and the bias estimate \eqref{eq:jensen-bias}; this is its
status in \Cref{tab:status}.
\end{remark}

\subsection{Verifying the deflated value}

For $\Jmin$ the value-estimation error and the subsolution-verification error
are distinct: a confidence interval for $\widetilde\psi_\lambda(x_0,t_0)$ says
nothing about the sign of $\tr(\Kdef\nabla^2\Jmin)$ on a domain. Directional
Hessian traces can be estimated by common-random-number finite differences or
stochastic differentiation, but a certificate needs an analytic sign theorem
such as \Cref{prop:prekopa}, a validated enclosure on the whole domain, or a
uniform negative bound followed by \Cref{lem:shift}. For the nonlinear
application only partial versions of the second are available;
\Cref{sec:app-deflation} records what was obtained and why a domain-wide
enclosure is currently out of reach.

\subsection{Relation to finite-dimensional searches}

Before minimization over $u$, the subsolution inequalities are linear in $\Jc$
and dual to occupation-measure linear programs \cite{lasserre2008nonlinear};
after it, the reduced inequality is convex but nonlinear in $\Jc$ and can be
written pointwise as a Schur-complement constraint. Sum-of-squares relaxations
give a finite-dimensional route on compact semialgebraic domains
\cite{parrilo2000structured,lasserre2008nonlinear,jones2020polynomial}. The
constructions of \Cref{sec:relaxations} are complementary: they produce a single
candidate by solving a surrogate exactly, rather than searching a parametrized
family, and their cost is that of Monte Carlo integration rather than of a
semidefinite program in the state dimension.

\section{Stochastic swing networks}
\label{sec:swing}

The results of \Cref{sec:certificates,sec:relaxations,sec:computing,sec:design}
hold for any plant of the form \eqref{eq:sde}. We now apply them to the power
grids studied in \Cref{sec:numerics}: a lossless, fixed-voltage transmission
network after a large load loss, with random power fluctuations at every bus
and fast active-power control at the generator buses only. This is exactly the
underactuated situation for which completion was designed --- the noise acts in
directions that no actuator reaches --- and the network structure lets every
object of \Cref{sec:relaxations} be written down explicitly rather than computed
abstractly. The section does four things. First, it fixes the model and the
performance objective and checks the standing assumptions of
\Cref{thm:completion}. Second, it computes the matching geometry: inertia
cancels, the admissibility condition $\lambda G\preceq Q$ becomes a bus-by-bus
comparison of control authority with noise variance, and the maximal matching
temperature is set by the worst-matched generator, \eqref{eq:lambda-power}.
Third, it identifies what each relaxation changes: completion adds virtual power
authority, at the load buses when the noise is homogeneous
(\Cref{cor:homogeneous}), and deflation deletes noise in such a way that its
curvature condition involves only the frequency directions, \eqref{eq:rho}.
Fourth, it converts the virtual authority into an equivalent power rating, so
that the planning results of \Cref{sec:design} can be read as statements about
installable flexibility. These are the quantities computed in
\Cref{sec:numerics}.

\subsection{Model}

Let a connected transmission graph have $n$ buses, $L$ energized branches,
incidence matrix $\Inc\in\R^{n\times L}$ and line strengths
$K=\operatorname{diag}(K_e)>0$. After a load-loss event and a prescribed
post-event balancing allocation, the controlled stochastic swing model is
\begin{equation}
 \dd\theta_t=\omega_t\dd t,\qquad
 M\dd\omega_t=\bigl[p^+-C\omega_t-\Inc K\sin(\Inc^\top\theta_t)+\Sg u_t\bigr]\dd t
 +\Gamma\dd W_t ,
 \label{eq:swing}
\end{equation}
with diagonal inertia $M\succ0$ and damping $C\succeq0$, independent nodal
injection noise $\Gamma=\operatorname{diag}(\Gamma_i)$, and $\Sg$ selecting
generator buses. Let $\theta^\star$ solve $p^+=\Inc K\sin(\Inc^\top\theta^\star)$
and write $\delta=\Inc^\top\theta$. On the cohesive cell
$\Omega_{\rm coh}=\{\theta:|\delta_e|<\pi/2\ \forall e\}$ the post-event potential
$U(\theta)=-p^{+\top}\theta-\sum_eK_e\cos\delta_e$ is convex modulo the uniform
phase, with Hessian $L(\theta)=\Inc\operatorname{diag}(K_e\cos\delta_e)\Inc^\top$,
and the relative energy
$\mathcal E=\tfrac12\omega^\top M\omega+U(\theta)-U(\theta^\star)$ is dissipated by
the uncontrolled deterministic flow at rate $\omega^\top C\omega$
\cite{dorfler2013synchronization,dorfler2018electrical}. The equilibrium and its
cohesive geometry are taken from \cite{chertkov2026cyclespace}.

\subsection{Performance objective and separate safety}
\label{sec:power-objective}

Persistent noise prevents convergence to a point, so we use a finite-horizon
objective rewarding concentration near $(\theta^\star,0)$:
\begin{equation}
 V=c_E\,\frac{\mathcal E_+}{n}+c_B\,\overline{\mathcal B}(\delta),\qquad
 \phi=c_T\,\frac{\mathcal E_+}{n},\qquad \mathcal E_+:=\max(\mathcal E,0).
 \label{eq:power-cost}
\end{equation}
The barrier $\overline{\mathcal B}$ is the edge average of the Bregman divergence
of $h(\delta)=-\log\cos\delta$ about $\delta^\star_e=(\Inc^\top\theta^\star)_e$,
kept exactly up to a guard angle of $82^\circ$ and continued beyond it by the
fourth-order Taylor polynomial of $h$, whose coefficients are all positive. The
objective therefore stays finite for trajectories that cross $90^\circ$. The
positive part $\mathcal E_+$ matters only outside $\Omega_{\rm coh}$: on the
cohesive cell $U$ is convex with minimum at $\theta^\star$, so $\mathcal E\ge0$
there. Synchrony safety is reported
\emph{separately}, through the first crossing $\tau_{\rm coh}=\inf\{t:\max_e
|\delta_e(t)|\ge\pi/2\}$ and its probability $\Prob^u(\tau_{\rm coh}\le T)$. The
separation is deliberate. An absorbing failure cost turns the Feynman--Kac
expectation into an extreme rare-survival problem near the boundary, whereas the
smooth value stays statistically regular (Figure~\ref{fig:sm-radial} in \cref{app:figures});
a certificate for the crossing probability is a different object, sketched in
\Cref{sec:app-firstexit}.

\begin{remark}[Standing assumptions in the application]
\label{rem:power-assumptions}
The noise acts on every frequency coordinate ($\Gamma_i>0$ for all $i$ in both
noise models below), so the H\"ormander condition of \Cref{rem:hypoelliptic}
holds: brackets of $e_{\omega_i}$ with the drift generate $e_{\theta_i}$. Because
$V,\phi\ge0$, \Cref{ass:standing} holds with $S_{\min}=0$, the Feynman--Kac
weights lie in $(0,1]$ under passive sampling, and \Cref{prop:lcb} applies in
its most favourable form (and is nonetheless vacuous, \Cref{rem:lcb-vacuous}).
The positive part makes $V$ and $\phi$ only Lipschitz across the set
$\{\mathcal E=0\}$, which lies outside $\Omega_{\rm coh}$ (the potential $U$ is
not periodic, and phase slips lower it). The smoothness of $\Jp$ supplied by
\Cref{rem:hypoelliptic} therefore holds away from that set; classical regularity
on the whole state space is assumed rather than proved, and we record this in
\Cref{tab:status}. In the reported samples the 99th percentile of the maximal
line angle is about $100$--$103^\circ$, far from a full slip.
\end{remark}

\subsection{Geometry of generator-only actuation}

In state coordinates $x=(\theta,\omega)$ with nodal actuator matrix $A_g$
($A_g=\Sg$ for independent generator control, $A_g=\Sg A$ for an orthonormal
generator-mode matrix $A$) and diagonal metric $R$,
\begin{equation}
 B=\begin{bmatrix}0\\M^{-1}A_g\end{bmatrix},\qquad
 \sigma=\begin{bmatrix}0\\M^{-1}\Gamma\end{bmatrix},
 \qquad
 G=\begin{bmatrix}0&0\\0&M^{-1}HM^{-1}\end{bmatrix},\quad
 Q=\begin{bmatrix}0&0\\0&M^{-1}\Gamma^2M^{-1}\end{bmatrix},
 \label{eq:power-geometry}
\end{equation}
where $H:=A_gR^{-1}A_g^\top$ is the nodal injection-control Gramian. $\sigma$ is
singular, but $\ran(G)\subseteq\ran(Q)$ whenever every actuated bus is noisy,
and $\lambda G\preceq Q$ is equivalent to $\lambda H\preceq\Gamma^2$:
\emph{inertia cancels from the matching geometry}. For independent generator
control with $R=\operatorname{diag}(r_g)$, \eqref{eq:diag-lambda} gives
\begin{equation}
 \lamstar=\min_{g}\,r_g\Gamma_g^2 .
 \label{eq:lambda-power}
\end{equation}

\paragraph{Completion}
In nodal coordinates the completion adds virtual injection authority
$H_v:=\Gamma^2/\lamstar-H\succeq0$; with $L_vL_v^\top=H_v$,
\begin{equation}
 M\dd\omega_t=\bigl[p^+-C\omega_t-\Inc K\sin(\Inc^\top\theta_t)+A_gu_t+L_vv_t\bigr]
 \dd t+\Gamma\dd W_t ,
 \label{eq:swing-completed}
\end{equation}
priced by $\tfrac12\|v\|^2$. By \Cref{thm:completion} its value is computed by
Feynman--Kac under the actual all-bus forcing, and
$\Res[\Jp]=\tfrac12(\nabla_\omega\Jp)^\top M^{-1}H_vM^{-1}\nabla_\omega\Jp\ge0$.
The diagnostics \eqref{eq:diagnostics} become
$\Ip=\tr(\Gamma^2/\lamstar-H)$ and
$\Dp=\tr(I-\lamstar\Gamma^{-1}H\Gamma^{-1})$.

\begin{corollary}[Homogeneous noise]
\label{cor:homogeneous}
If $\Gamma=\sigma_0I$, $R=rI$ and generators are independently controlled,
then $\lamstar=r\sigma_0^2$ and $H_v=r^{-1}(I-\Sg\Sg^\top)$: the least completion
adds virtual power authority at \emph{load buses only}, and the generator
directions are already exactly matched. For orthonormal actuation of rank
$r_a$, $\Dp=n-r_a$.
\end{corollary}

Under heterogeneous noise the picture changes. By \eqref{eq:lambda-power} one
weak-noise generator sets the temperature for all, so adding an actuator on a
quiet bus can \emph{lower} $\lamstar$ and \emph{raise} $\Ip$, even though it
improves the physical problem. We call this the \emph{weakest-channel effect};
it is the mechanism behind the numerical results of \Cref{sec:numerics}, and
\Cref{prop:noise-proportional} is its remedy.

\paragraph{Deflation}
At $\lamstar$ the deleted diffusion consists of the load-bus noise plus any
unmatched generator surplus, and \Cref{thm:deflation} reduces to
\begin{equation}
 \rho(x,t):=\Res[\Jmin]=\tr\bigl[\Kdef_{\omega\omega}\,\nabla^2_{\omega\omega}\Jmin(x,t)\bigr].
 \label{eq:rho}
\end{equation}
Curvature in the rotor angles is irrelevant; only curvature of the value in the
deleted \emph{frequency} directions matters. Verification of $\rho\ge0$ is
discussed in \Cref{sec:app-deflation}.

\paragraph{From virtual to physical authority}
For planning, diagonal or low-rank pieces of $H_v$ can represent installable
signed active-power flexibility. A device whose dimensionless command $a_i$ maps
to power $P_ia_i$ at cost $\tfrac12r_ia_i^2$ contributes $h_i=P_i^2/r_i$ to the
nodal Gramian. We use $P_i^{\rm eff}=S_{\rm base}\sqrt{r_{\rm ref}h_i}$ as an
\emph{authority proxy} in megawatts. It is a Gramian-to-power conversion, not a
device rating: energy capacity, state of charge, ramp limits and degradation are
separate constraints that break exact matching and are not modelled here.

\section{Numerical study}
\label{sec:numerics}

We certify simple generator controllers on the IEEE 39-bus (New England) and
118-bus systems \cite{zimmerman2011matpower} after a severe load-loss event. The
39-bus study exposes the mechanism; the 118-bus study is the scaling test and
supplies the headline numbers. \Cref{tab:headline} answers the question in the
title for both.

\subsection{Modelling scope and protocol}
\label{sec:scope}

The following choices hold throughout and are stated once. The network is the
lossless, fixed-voltage reduction of \Cref{sec:swing}; load buses carry a small
synthetic inertia ($M_i=0.02$, against $0.30$ at generators), which makes the
model stiff. Injections are scaled by $\alpha=4$, a deliberately stressed
laboratory rather than a calibrated operating point. The event is inherited from
\cite{chertkov2026cyclespace} and the post-event balancing is equal across
generators. Noise is white and independent across buses, in two spatial
profiles of equal root-mean-square amplitude $\sigma_{\rm rms}$:
\emph{homogeneous}, $\Gamma_i=\sigma_{\rm rms}$; and \emph{activity-scaled},
$\widehat\Gamma_i=0.25+\min\bigl(a_i/q_{0.75}(a),4\bigr)$ with $a_i=|p_i|$ and
$q_{0.75}$ the upper quartile, renormalized to the same RMS. The second is a
stylized proxy for fluctuations that scale with nodal injection; it is not
fitted to data. The deployed controllers are deliberately simple---static
balancing, generator droop, and a scaled linear-quadratic (LQR) feedback about
the post-event equilibrium, with one scalar gain chosen on training paths
disjoint from the evaluation paths---because the point is to certify a
controller that was not designed from the certificate. The horizon is $T=3$\,s,
$R=0.2I$ in orthonormal actuator coordinates, and $(c_E,c_B,c_T)=(0.08,0.02,0.20)$
in \eqref{eq:power-cost}. The smooth objective replaces an absorbing failure
cost because, near the cohesion boundary, the Feynman--Kac estimator of the
stopped formulation is dominated by rare surviving paths while the smooth one
stays well conditioned (Figure~\ref{fig:sm-radial} in \cref{app:figures}). Paired comparisons use common random numbers; time
steps are $\Delta t\le5\times10^{-3}$\,s (case39) and $2.5\times10^{-3}$\,s
(case118), below the threshold at which the stiff model produces spurious angle
excursions (Figure~\ref{fig:sm-convergence} in \cref{app:figures}). Governor and voltage
dynamics, losses, protection and actuator energy limits are absent.

Four actuator classes are compared in both systems. H0 is a single physical
slack generator; H1 is one distributed-slack generator mode; H2 is a rank-three
generator subspace spanned by the uniform mode and the two softest nonconstant
modes of the Kron-reduced post-event generator stiffness; H3 is independent
control of every generator. H4 is not physical: it is the completion of H3,
used as a common oracle.

\begin{table}[t]
\caption{How suboptimal is the deployed controller? Certified gap
$(J-\Jp)/J$ of the H3 scaled-LQR controller, where $J$ is the deployed cost and
$\Jp$ the completed lower value (\Cref{thm:completion}). For case118,
$\pm$ is one standard error of the deployed-cost estimate; the completed values
have replicate standard errors below $2\times10^{-4}$ in absolute terms.
Repricing changes $R$, hence the control problem (\Cref{rem:repricing}).}
\label{tab:headline}
\centering\footnotesize\setlength{\tabcolsep}{5pt}
\begin{tabular}{llllr}
\toprule
system & noise & pricing & $\sigma_{\rm rms}$ & certified gap \\
\midrule
case39  & homogeneous    & uniform             & 0.10 & $7.5\%$ \\
case39  & activity-scaled & uniform            & 0.10 & $32.5\%$ \\
case39  & activity-scaled & noise-proportional & 0.10 & $3.5\%$ \\
\midrule
case118 & homogeneous    & uniform             & 0.15 & $2.8\pm0.5\%$ \\
case118 & homogeneous    & robust fixed        & 0.15 & $4.8\pm0.5\%$ \\
case118 & activity-scaled & uniform            & 0.15 & $39.1\pm0.8\%$ \\
case118 & activity-scaled & robust fixed       & 0.15 & $5.3\pm0.8\%$ \\
case118 & activity-scaled & noise-proportional & 0.15 & $1.2\pm0.2\%$ \\
\bottomrule
\end{tabular}
\end{table}

\subsection{The 39-bus system: mechanism}
\label{sec:case39}

\paragraph{Event}
The lossless reduction has 39 buses, 46 energized branches and 10 generators.
At $\alpha=4$ a $2724.5$\,MW load block at bus 20 disappears; the pre-event
maximum line angle is $46.17^\circ$, and with equal final balancing the
deterministic transient peaks at $81.11^\circ$ and stays cohesive
(\Cref{fig:case39-event}, left). Concentrating the final balancing on fewer
generators makes the event unsafe: at nominal noise the empirical crossing
probability is $0.59$ at concentration $0.25$ and essentially one for
concentrations $\ge0.5$ (\Cref{fig:case39-event}, right). Equal balancing is
therefore used as the common target for all controllers.

\begin{figure}[t]
  \centering
  \includegraphics[width=0.49\textwidth]{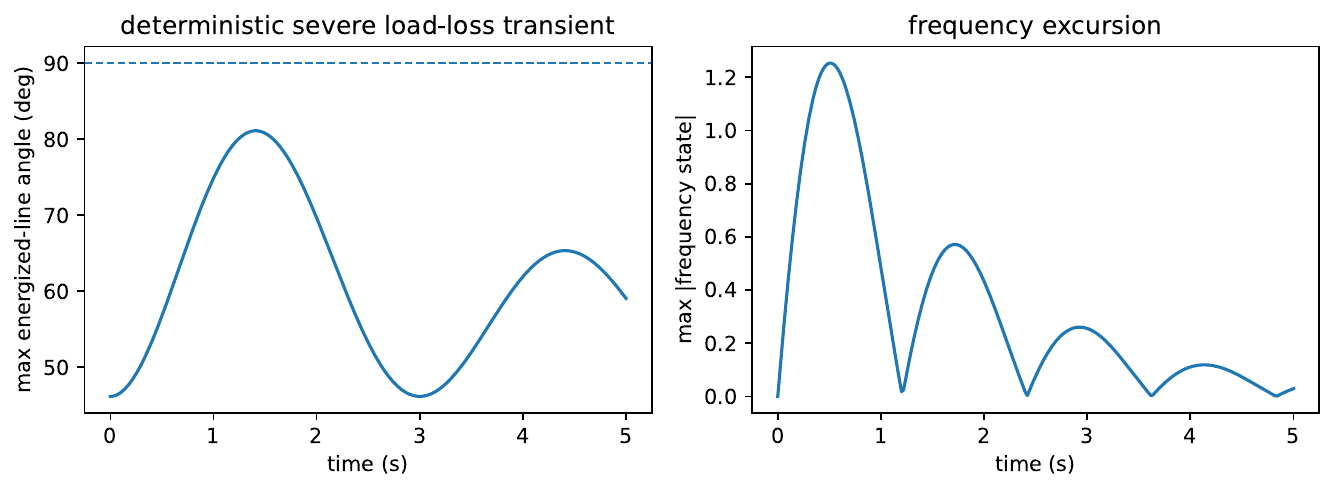}\hfill
  \includegraphics[width=0.49\textwidth]{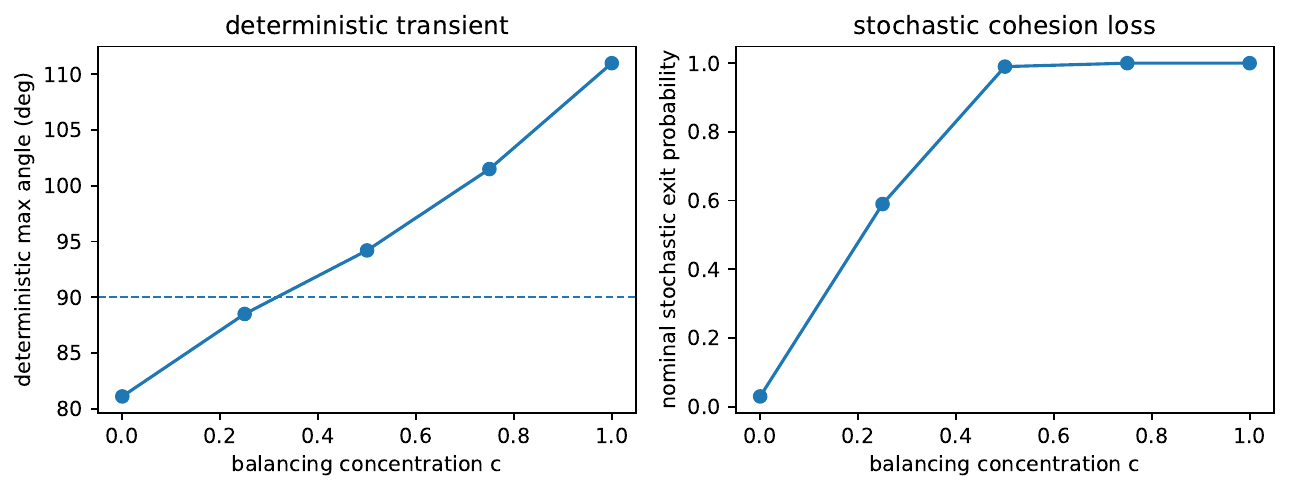}
  \caption{The 39-bus load-loss event. Left: deterministic transient with equal
  final balancing; the maximum line angle peaks near $81^\circ$ and the
  trajectory stays cohesive. Right: concentrating the post-event balancing
  raises the deterministic peak angle and the empirical probability of crossing
  $90^\circ$ under nominal noise.}
  \label{fig:case39-event}
\end{figure}

\paragraph{Deployed controllers}
With droop gain $0.025$ and H3 LQR scale $0.01$, the empirical probabilities of
crossing $90^\circ$ at homogeneous $\sigma_{\rm rms}=0.2$ are about $0.129$,
$0.098$ and $0.094$ for static balancing, droop and H3 LQR, rising to $0.375$,
$0.301$ and $0.285$ at $\sigma_{\rm rms}=0.3$. Noise placement matters in
addition to amplitude (Figure~\ref{fig:sm-noise} in \cref{app:figures}). As an independent
check, the PIC feedback estimator projected onto the full H3 LQR command gives a
coefficient close to $0.0100$ (homogeneous) and $0.0097$ (activity-scaled) away
from the ill-conditioned $t=0$ projection, consistent with the gain selected by
the upper-cost scan (Figure~\ref{fig:sm-projection} in \cref{app:figures}).

\paragraph{Completion across the actuator hierarchy}
\Cref{tab:case39} and \Cref{fig:case39-hierarchy} compare the four
architectures at nominal noise. For homogeneous noise, $\lamstar=r\sigma^2$ is
the same for all four (\Cref{cor:homogeneous}), so the completed process is
identical and $\Jp$ does not change with the architecture: more physical rank
simply replaces fictitious directions by real ones. The deflated values of
\Cref{tab:case39}, with the rank of the deleted stochastic subspace, are
plotted in Figure~\ref{fig:sm-h0h3} in \cref{app:figures}. The completion burden falls
from $190$ (H0, H1) to $180$ (H2) and $145$ (H3), the whitened deficit is
exactly $38,38,36,29$, and the gap $J-\Jp$ falls from $1.1\times10^{-2}$ to
$6.9\times10^{-3}$. Here the geometry behaves as intuition about
underactuation suggests.

\begin{table}[t]
\caption{Case39 actuator hierarchy at $\sigma_{\rm rms}=0.1$. $J$: deployed
cost; $\Jmin$: deflated value (a lower bound only if the curvature condition of
\Cref{thm:deflation} holds); $\Jp$: completed value (an automatic lower bound);
$\Ip$: completion burden. All entries are Monte Carlo estimates.}
\label{tab:case39}
\centering\small
\begin{tabular}{llrrrrr}
\toprule
noise & class & $J$ & $\Jmin$ & $\Jp$ & $\Ip$ & $J-\Jp$ \\
\midrule
homogeneous & H0 & .09573 & .05770 & .08465 & 190.0 & .01108 \\
            & H1 & .09584 & .05778 & .08465 & 190.0 & .01120 \\
            & H2 & .09270 & .05511 & .08465 & 180.0 & .00805 \\
            & H3 & .09156 & .05619 & .08465 & 145.0 & .00691 \\
\midrule
activity-scaled & H0 & .08301 & .05816 & .07846 & 62.4  & .00455 \\
                & H1 & .08314 & .05779 & .07207 & 166.8 & .01107 \\
                & H2 & .08004 & .05501 & .06882 & 241.3 & .01122 \\
                & H3 & .07878 & .05414 & .05318 & 744.5 & .02560 \\
\bottomrule
\end{tabular}
\end{table}

For activity-scaled noise the picture reverses. The matching temperatures for
H0--H3 are
\[
 \lamstar\approx(5.79,\ 2.27,\ 1.52,\ 0.491)\times10^{-3}:
\]
each added direction exposes a quieter generator channel, and the single
temperature must match all of them. The completion burden therefore \emph{grows},
$\Ip\approx(62.4,166.8,241.3,744.5)$, although the whitened deficit falls only
mildly, from $38$ to $36.84$. This is the weakest-channel effect of
\Cref{sec:swing}: more physical rank improves the real problem but makes the
scalar completion more optimistic. Across the eight points of
\Cref{tab:case39}, $\Ip$ tracks the gap $J-\Jp$ and the whitened deficit does
not; we read this as evidence for $\Ip$ as a diagnostic, not as a theorem. Over
$\sigma_{\rm rms}=0.05,0.1,0.2,0.3$ the H3 completed gap is
$0.0040,0.0069,0.0178,0.0313$ (homogeneous) and $0.0205,0.0256,0.0428,0.0733$
(activity-scaled), while the geometry itself is invariant under
$\Gamma\mapsto s\Gamma$ (Figure~\ref{fig:sm-inflation-noise} in \cref{app:figures}). At the
weakest heterogeneous noise the two relaxations change order, as
\Cref{rem:not-ordered} allows.

\begin{figure}[t]
  \centering
  \includegraphics[width=0.92\textwidth]{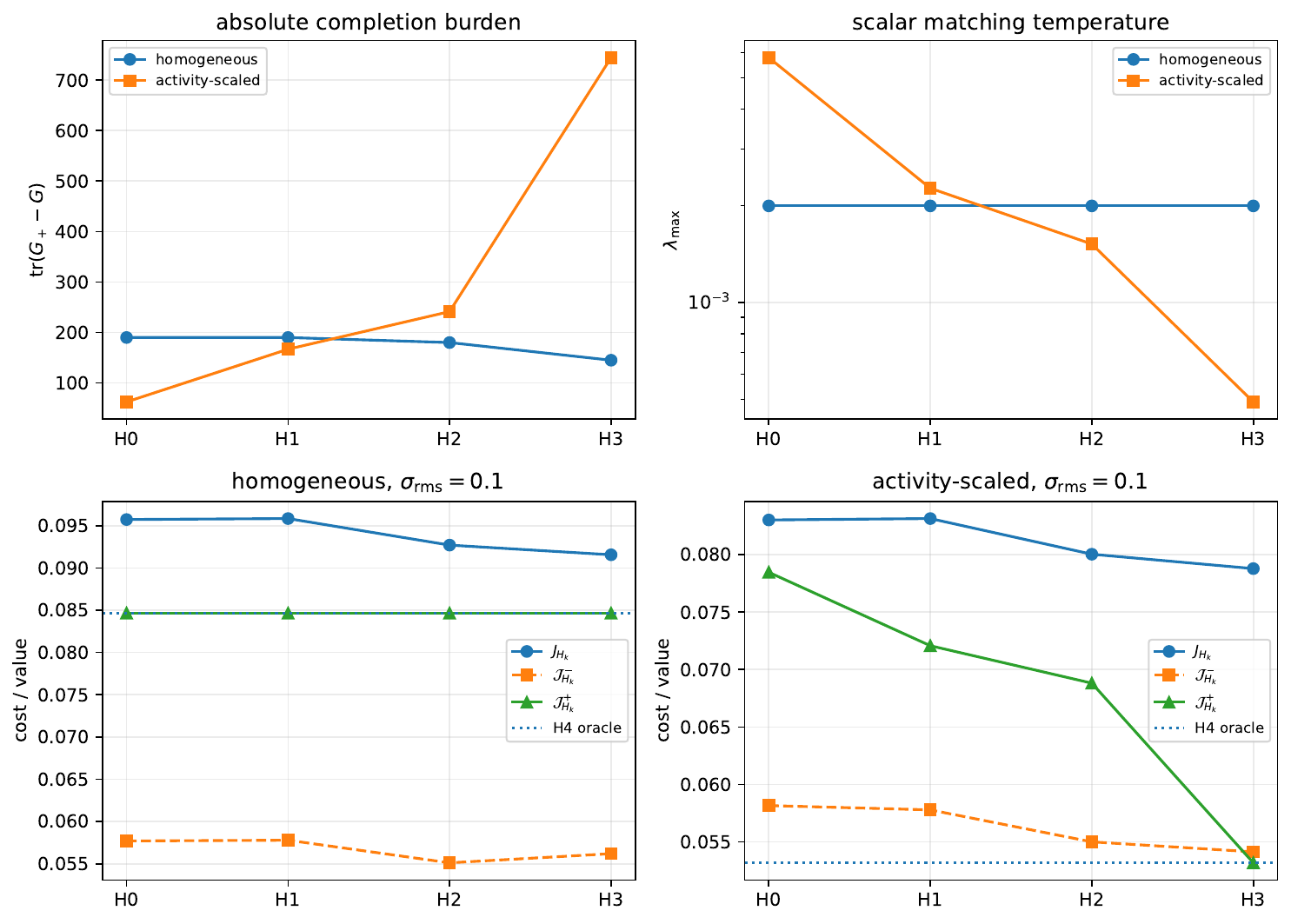}
  \caption{Case39 completion across the actuator hierarchy at
  $\sigma_{\rm rms}=0.1$. Top: completion burden $\tr(G_+-G)$ and matching
  temperature $\lamstar$. Bottom: deployed cost $J$, deflated value $\Jmin$,
  completed value $\Jp$, and the H4 oracle. Under homogeneous noise every
  architecture has the same completion; under activity-scaled noise the
  weakest-channel effect separates them.}
  \label{fig:case39-hierarchy}
\end{figure}

\paragraph{Pricing before hardware}
At fixed total generator authority $\sum_gg_g=50$, homogeneous noise is already
balanced, so uniform and noise-proportional pricing coincide. For
activity-scaled noise, \Cref{prop:noise-proportional} changes $\lamstar$ from
$4.91\times10^{-4}$ to $4.43\times10^{-3}$ and $\Ip$ from $744.5$ to $38.0$,
without adding authority; $\Jp$ rises from $0.0530$ to $0.0773$ against a
deployed cost of about $0.0801$ (Figure~\ref{fig:sm-metric} in \cref{app:figures}). The three
allocations compared are shown in Figure~\ref{fig:sm-allocation} in \cref{app:figures}. The robust allocation of \Cref{prop:robust} has
completion trace $169.3$ in both scenarios, with $\Jp=0.0834$ (homogeneous) and $0.0691$ (activity-scaled)
(Figure~\ref{fig:sm-scenario} in \cref{app:figures}). Moving down the linearly solvable
frontier \eqref{eq:frontier} from H3 at homogeneous noise, $\lambda/\lamstar=0.8,
0.6,0.4,0.25$ raises the trace from $145$ to about $194,275,438,730$ and lowers
$\Jp$ from $0.0846$ to $0.0823,0.0789,0.0738,0.0693$, with the authority proxy
rising from about $2.9$ to $7.5$\,GW (Figure~\ref{fig:sm-frontier39} in \cref{app:figures}).
Adding only $k=5$ load-bus devices ranked by an H4 shadow proxy lowers the
homogeneous deployed cost to $0.08887\pm0.00047$, against $0.09025\pm0.00048$
for the same number ranked by local noise and about $0.091$ for generators only
(Figure~\ref{fig:sm-sparse39} in \cref{app:figures}).

\paragraph{Implementation defects}
Imposed on the H3 controller without redesign, a $0.2$\,s sample-and-hold
interval raises the cost by about $2.1\times10^{-3}$ and a $0.2$\,s first-order
lag by about $2.2\times10^{-3}$, roughly $2\%$ of the objective; saturation is
negligible until the cap falls well below the 99th-percentile command, and
generator dropout raises the mean cost monotonically (\Cref{fig:case39-impl}).
These are the defects of \Cref{sec:implementation}, priced against the same
objective without a new HJB solve.

\begin{figure}[t]
  \centering
  \includegraphics[width=0.78\textwidth]{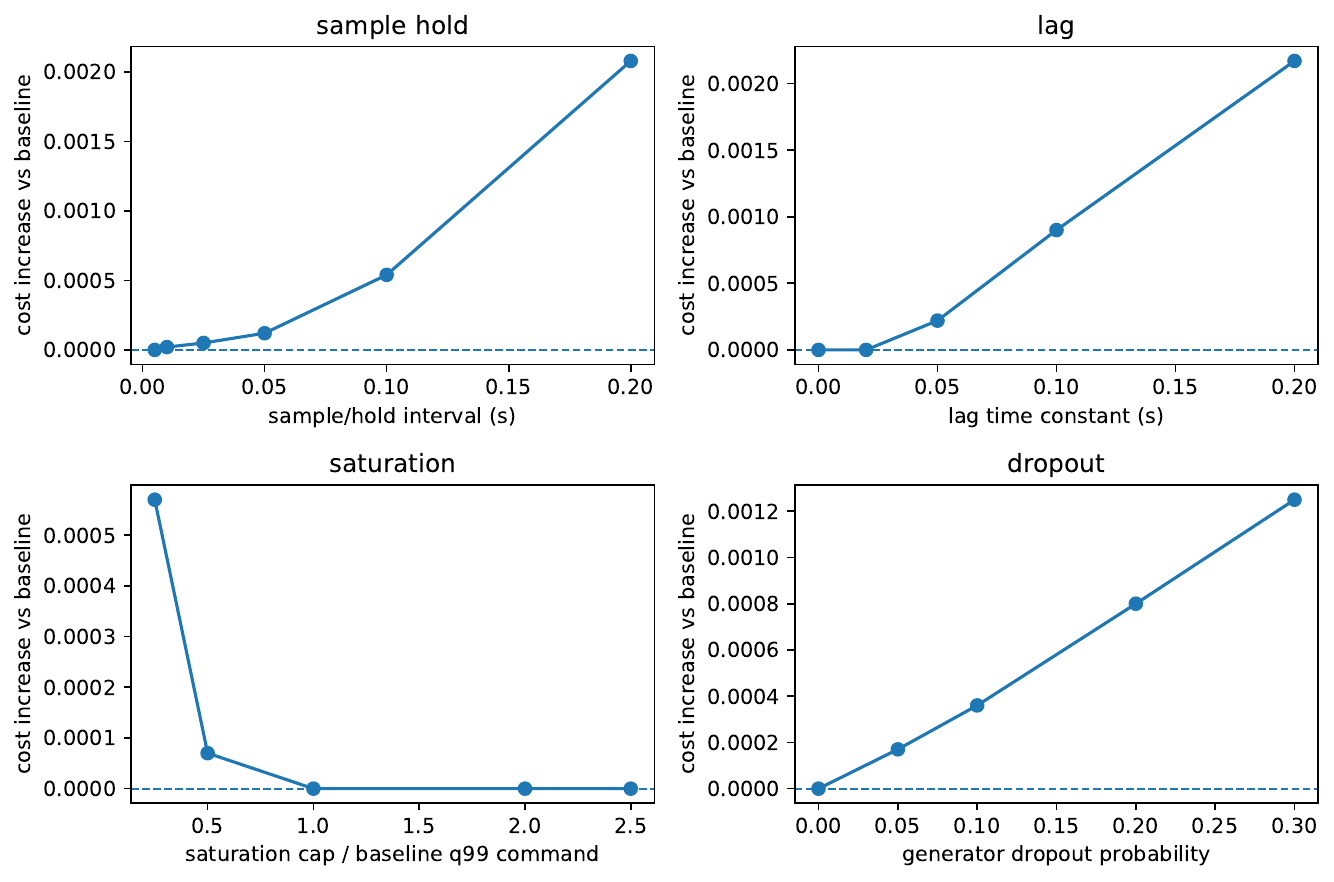}
  \caption{Case39 implementation defects imposed on the nominal H3 controller:
  sample-and-hold interval, first-order actuator lag, saturation cap, and
  event-level generator dropout. Each panel shows the cost increase relative to
  the ideal implementation.}
  \label{fig:case39-impl}
\end{figure}

\subsection{The 118-bus system: scaling}
\label{sec:case118}

\paragraph{Event}
The reduction has 118 buses, 54 generator buses and 186 branches, with bus 1
fixing the gauge; its deterministic thresholds reproduce those of
\cite{chertkov2026cyclespace}. At $\alpha=4$ the complete net load at bus 90,
$656.6$\,MW, is lost. The maximum line angle is $59.9^\circ$ before the event and
$63.3^\circ$ at the post-event equilibrium, and the uncontrolled deterministic
transient peaks at $72.8^\circ$ over five seconds.

\paragraph{Deployed controllers}
H0 is the slack generator at bus 69. Ranks $3$, $5$ and $10$ were compared for
H2 and were indistinguishable at the training resolution, so rank three was
kept. The selected LQR scales are $0.00125,0.00125,0.0025,0.005$ for H0--H3.
Across $\sigma_{\rm rms}=0.10,0.15,0.20$, H3 LQR lowers the smooth cost relative
to static balancing by $0.5$--$0.8\%$ (homogeneous) and $1.3$--$1.6\%$
(activity-scaled); H0 and H1 are indistinguishable from static balancing, so
actuator rank is the dominant distinction in this weak-feedback regime
(\Cref{tab:case118-deployed}, Figure~\ref{fig:sm-deployed118} in \cref{app:figures}). The
empirical crossing probabilities in \Cref{tab:case118-deployed} are safety
diagnostics outside the smooth objective.

\begin{table}[t]
\caption{Case118 deployed policies at $\sigma_{\rm rms}=0.15$ (192 evaluation
paths; standard error of $J$ about $3.5\times10^{-4}$ homogeneous and
$3.2\times10^{-4}$ activity-scaled).}
\label{tab:case118-deployed}
\centering\small
\begin{tabular}{llrrr}
\toprule
noise & policy & $J$ & $\Prob(\tau_{\rm coh}\le T)$ & $q_{0.99}$ max angle \\
\midrule
homogeneous     & static   & 0.077150 & 0.328 & $103.3^\circ$ \\
homogeneous     & H2 droop & 0.077062 & 0.297 & $102.0^\circ$ \\
homogeneous     & H3 LQR   & 0.076692 & 0.292 & $102.0^\circ$ \\
activity-scaled & static   & 0.041297 & 0.208 & $101.5^\circ$ \\
activity-scaled & H2 droop & 0.041222 & 0.182 & $100.5^\circ$ \\
activity-scaled & H3 LQR   & 0.040731 & 0.167 & $100.1^\circ$ \\
\bottomrule
\end{tabular}
\end{table}

\paragraph{Completion across the hierarchy}
For homogeneous noise $\lamstar=0.2\sigma_{\rm rms}^2$ for all four
architectures, so the completed process is architecture-independent
(\Cref{tab:case118}). For activity-scaled noise the matching temperatures are
$(3.65\times10^{-2},1.31\times10^{-3},1.09\times10^{-3},1.76\times10^{-4})$ for
H0--H3, and the H3 completion trace reaches $1.48\times10^4$: the
weakest-channel effect again, now much stronger. The hard activity-scaled H3
completion was audited separately, because its first importance sampler had an
effective sample fraction of about $0.011$; a six-bin time-dependent guide
raised it to $0.029$ and gave $\Jp\approx0.0250$ with replicate standard error
$1.7\times10^{-5}$ over $2048$ paths. By \Cref{lem:jensen} its estimated bias is
about $6\times10^{-6}$, below the standard error and three orders of magnitude
below the gap; the unquantified risk is weight mass that was never sampled
(\Cref{rem:lcb-vacuous}).

\begin{table}[t]
\caption{Case118 hierarchy at $\sigma_{\rm rms}=0.15$. $\Jmin$ is a lower value
only under the curvature condition of \Cref{thm:deflation}, which is not
verified here; $\Jp$ is an exact lower value estimated by Monte Carlo.
The small negative entry is within half a standard error of the deployed-cost
estimate, whose standard error is $3.2\times10^{-4}$; the completed estimate at
that point has effective sample fraction $0.999$, so it is not a bias effect.}
\label{tab:case118}
\centering\small
\begin{tabular}{llrrrr}
\toprule
noise & class & $J$ & $\Jmin$ & $\Jp$ & $J-\Jp$ \\
\midrule
homogeneous & H0 & 0.077146 & 0.007642 & 0.075024 & $0.002122$ \\
            & H1 & 0.077146 & 0.007575 & 0.075024 & $0.002122$ \\
            & H2 & 0.077080 & 0.007925 & 0.075024 & $0.002056$ \\
            & H3 & 0.076692 & 0.018356 & 0.075024 & $0.001668$ \\
\midrule
activity-scaled & H0 & 0.041286 & 0.009109 & 0.041440 & $-0.000153$ \\
                & H1 & 0.041293 & 0.007476 & 0.036321 & $0.004972$ \\
                & H2 & 0.041234 & 0.007511 & 0.035570 & $0.005664$ \\
                & H3 & 0.040731 & 0.007650 & 0.024957 & $0.015774$ \\
\bottomrule
\end{tabular}
\end{table}

In this system the completed value is the tighter lower bound as well as the
rigorous one: $\Jmin$ is below $\Jp$ by a factor between three and ten in every
row. We therefore base all case118 certificates on $\Jp$.

\paragraph{Pricing existing authority}
\Cref{fig:case118-certificate} and \Cref{tab:headline} show the central result.
All pricings keep the total generator authority fixed at $\sum_gg_g=270$, and the
deployed LQR is re-solved for each pricing with one retuned scalar gain.
Activity-specific pricing (\Cref{prop:noise-proportional}) raises the
activity-scaled matching temperature by a factor of about $50$, from
$1.76\times10^{-4}$ to $8.72\times10^{-3}$, and lowers the completion burden from
$14832$ to $34.5$. The certified gap falls from $39.1\%$ to $1.2\%$; the deployed
cost barely moves ($0.04095$ to $0.04105$, the latter from an independent
$3072$-path audit), so the change is in the lower bound. The robust pricing of
\Cref{prop:robust} gives $4.8\%$ and $5.3\%$ in the two scenarios, while
activity-specific pricing is deliberately mismatched to homogeneous noise
($11.6\%$). The gap tracks the completion burden across all six cases. Following
\Cref{rem:repricing}, the correct reading is the design statement: with prices
proportional to the noise, the simple controller is certified within about
$1\%$ of optimal. Whether the $39\%$ under uniform prices reflects a loose bound
or a poor controller cannot be decided from these data, because the optimum of
the uniform-price problem is unknown.

\begin{figure}[t]
  \centering
  \includegraphics[width=0.95\textwidth]{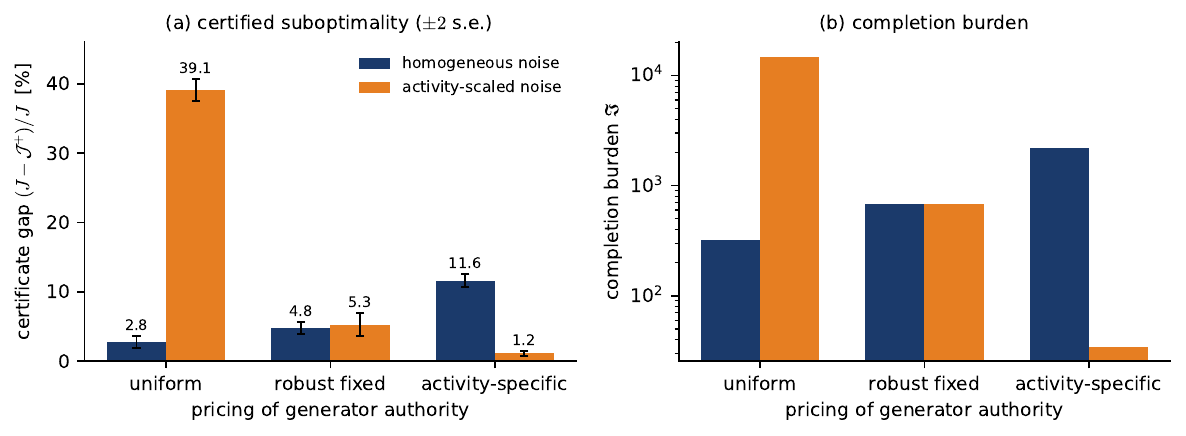}
  \caption{Case118 certificates under three pricings of the same total generator
  authority, $\sigma_{\rm rms}=0.15$. (a) Certified gap $(J-\Jp)/J$ of the H3
  scaled-LQR controller, with $\pm2$ standard errors of the deployed-cost
  estimate. (b) Completion burden $\Ip$ (log scale). The gap tracks the burden;
  noise-proportional pricing removes almost all of it for the noise shape it is
  matched to.}
  \label{fig:case118-certificate}
\end{figure}

\paragraph{The linearly solvable frontier}
\Cref{fig:case118-frontier} follows \eqref{eq:frontier} below $\lamstar$. For
homogeneous noise with uniform pricing, reducing $\lambda/\lamstar$ from $1$ to
$0.25$ lowers $\Jp$ from $0.07502$ to $0.06884$ ($8.3\%$) while the authority
proxy grows from $6.4$ to $22.2$\,GW. For activity-scaled noise with
activity-specific pricing, $\Jp$ falls from $0.04057$ to $0.03821$ ($5.8\%$) while
the proxy grows from $1.9$ to $11.3$\,GW; near the least completion the frontier
is nearly flat, the first $20\%$ temperature reduction adding $2.4$\,GW for a
change in $\Jp$ of $3\times10^{-5}$, comparable to the sampling resolution. At
the matched least completion the generator completion is zero and only the $64$
non-generator buses are virtual; below it, authority is needed in all $118$
directions. Exact linear solvability is therefore expensive as a construction
target, which motivates sparse reinforcement.

\begin{figure}[t]
  \centering
  \includegraphics[width=0.62\textwidth]{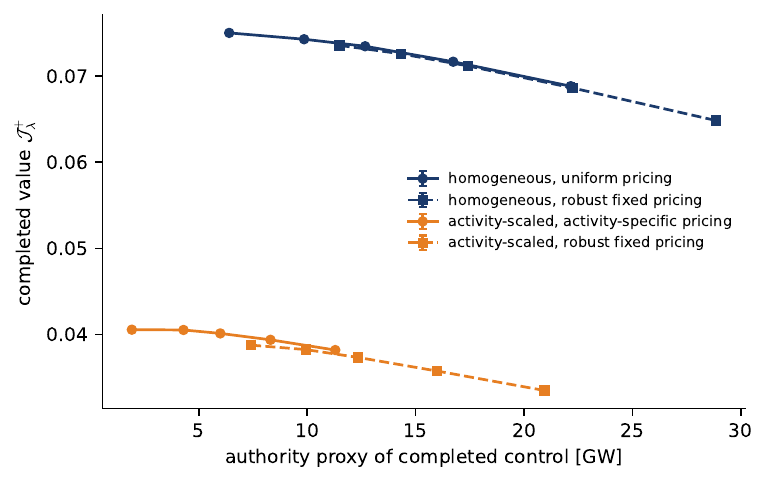}
  \caption{Case118 linearly solvable reinforcement frontiers: completed value
  $\Jp_\lambda$ against the authority proxy of the completed control as
  $\lambda/\lamstar$ decreases from $1$ to $0.25$ (error bars: $\pm2$ replicate
  standard errors). The proxy is a Gramian-to-power conversion, not a device
  rating.}
  \label{fig:case118-frontier}
\end{figure}

\paragraph{Sparse reinforcement}
Keeping the least completion as an oracle, we add real fast-flexibility devices
at $k=5,10,20,40,64$ load buses, with $k=64$ recovering the full least completion,
and compare two placement rules at exactly matched added authority: local noise
variance, and the shadow value of \Cref{prop:shadow} computed from the H4 LQR
closed loop. Under homogeneous noise the noise ranking is degenerate---every bus has the
same variance, so it is an arbitrary ordering---and the shadow rule is not
better: at $k=20$ and $40$ the arbitrary ordering gives slightly larger
reductions ($0.59\%$ against $0.54\%$, and $1.20\%$ against $1.15\%$), a difference
well resolved by the paired standard errors. Full $k=64$ completion lowers the
deployed cost by $1.97\%$ and the empirical crossing frequency from $0.302$ to
$0.250$. Under activity-scaled noise the two rules select disjoint top-five
buses, and the shadow rule has lower cost at $k=5,10,20,40$ with every paired
common-random-number $95\%$ interval excluding zero; it recovers about $8$, $35$,
$70$ and $79\%$ of the $k=64$ improvement, against a noise rule that is slightly
worse than no devices at $k=5,10$ and recovers $47\%$ at $k=40$
(\Cref{fig:case118-sparse}). The absolute effect is small: the entire $k=64$
improvement is $0.09\%$ of the deployed cost, and the crossing frequency is
$0.167$ for every sparse policy. Relative to the completed optimum $0.04057$, the
simple LQR closes only a small part of the remaining gap even with all $64$ load
directions physically present; the oracle reveals value this controller does not
extract.

\begin{figure}[t]
  \centering
  \includegraphics[width=0.95\textwidth]{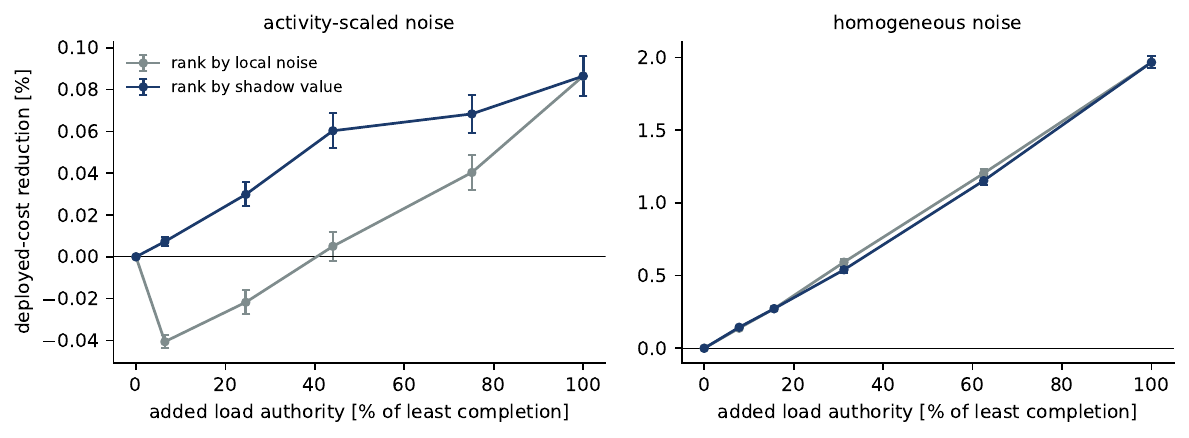}
  \caption{Case118 sparse load-side reinforcement at matched added authority:
  reduction of the deployed H3 cost relative to generator-only control, for
  placement by local noise variance and by shadow value
  (\Cref{prop:shadow}), with $\pm2$ paired standard errors. Note the different
  vertical scales: under activity-scaled noise the ranking is statistically
  resolved but the absolute effect is below $0.1\%$.}
  \label{fig:case118-sparse}
\end{figure}

\section{Discussion}
\label{sec:discussion}

\subsection{What is certified, and how}

\Cref{tab:status} collects the status of every claim. The distinction that
matters most is between the completed value $\Jp$, which is an exact lower
bound whose \emph{numerical value} is a Monte Carlo estimate, and the deflated
value $\Jmin$, which is a lower bound only if a curvature condition holds that
we have not verified on the whole domain.

\begin{table}[t]
\caption{Status of the claims of the paper.}
\label{tab:status}
\centering\small\setlength{\tabcolsep}{4pt}
\begin{tabular}{>{\raggedright\arraybackslash}p{0.43\textwidth}
                >{\raggedright\arraybackslash}p{0.23\textwidth}
                >{\raggedright\arraybackslash}p{0.25\textwidth}}
\toprule
claim & status & where \\
\midrule
performance identity; two-sided certificate; regret identity
  & theorem & \Cref{prop:identity,thm:certificate,cor:regret} \\
completed value is a subsolution, $\Jp\le J^\star$; exact gap split
  & theorem, under \cref{ass:completed} & \Cref{thm:completion,cor:gap-split} \\
least inflation is tightest in the scalar family
  & theorem & \Cref{prop:least-inflation} \\
deflation residual is directional curvature; affine--convex case
  & theorem & \Cref{thm:deflation,prop:prekopa} \\
noise-proportional and robust pricing; multi-scenario obstruction
  & theorem & \Cref{prop:noise-proportional,prop:robust,prop:multiscenario} \\
shadow price of authority
  & theorem, under regularity & \Cref{prop:shadow} \\
estimator bias sign and size; finite-sample lower bound
  & theorem; bound vacuous at low $\lambda$ & \Cref{lem:jensen,prop:lcb,rem:lcb-vacuous} \\
\addlinespace
regularity of $\Jp$ in the swing application
  & smooth away from $\{\mathcal E=0\}$; global regularity assumed
  & \Cref{rem:hypoelliptic,rem:power-assumptions} \\
numerical values of $\Jp$ and of the certificate gaps
  & Monte Carlo estimates of exact quantities, with standard errors
  & \Cref{tab:headline,tab:case39,tab:case118} \\
$\Jmin$ as a lower bound in the swing application
  & conditional: curvature verified only on a terminal layer and a finite set of states
  & \Cref{sec:app-deflation} \\
$\Ip$ as a predictor of the gap
  & empirical diagnostic & \Cref{sec:case39,sec:case118} \\
crossing probabilities $\Prob(\tau_{\rm coh}\le T)$
  & empirical; not certified & \Cref{tab:case118-deployed}, \Cref{sec:app-firstexit} \\
price of linearisability $\Pi$
  & defined, not computed & \Cref{sec:design} \\
\bottomrule
\end{tabular}
\end{table}

\subsection{Relation to Bound-Optimized Task Choice}
\label{sec:botc}

Anderson, Das and Tanaka \cite{anderson2026botc} recently proposed
Bound-Optimized Task Choice (BOTC) for problems that fail the path-integral
matching condition. A BOTC \emph{task} is obtained by PSD increases of the noise
covariance and/or of the control price until the modified problem is matched.
Both moves make the task harder, every valid task yields an \emph{upper} bound
on the physical cost-to-go, and BOTC optimizes over the task family to tighten
it. The approach is complementary to ours: \Cref{fig:geometry} shows the four
elementary moves relative to the physical pair $(Q,G)$. BOTC's two moves reach
the steepest ray $q=\lambda^+g$; completion and deflation reach the lowest ray
$q=\lamstar g$.

In the common-support, nonsingular case this has a spectral form. With\\
$A:=G^{-1/2}QG^{-1/2}$, the least deflation and the least completion both use
$\lamstar=\lambda_{\min}(A)$, since $\lamstar G\preceq Q$; the covariance-only
BOTC extreme must satisfy $Q\preceq\lambda^+G$, whose least value is
$\lambda^+=\lambda_{\max}(A)$. The two families touch the physical problem from
opposite spectral edges and, when both apply, give the envelope
\begin{equation}
 \Jp\ \le\ J^\star\ \le\ J_{\rm BOTC}.
 \label{eq:envelope}
\end{equation}
The support question is not a technicality. The BOTC formulation assumes that
the noise lies in the directly actuated subspace and that the control matrix
spans it. Our application violates this deliberately: load-bus noise has no
corresponding actuator, so $\lambda^+=\infty$ (the open markers in
\Cref{fig:geometry}). Completion resolves the mismatch by enlarging the actuator
support, deflation by restricting the stochastic support; a literal application
of BOTC to the all-bus-noise, generator-only-control problem would require an
extension of its task class. Once the support issue is resolved, the natural
programme is to optimize hard and easy linearly solvable surrogates jointly and
report the bracket \eqref{eq:envelope}, possibly reusing trajectories across task
choices by change of measure as BOTC does.

\subsection{Relation to HJB bounds and sum-of-squares}

Using sub- and supersolutions to bound value functions is broader than
path-integral control, and viscosity and stochastic Perron constructions give a
general foundation \cite{bayraktar2013perron}. Closer to the present setting,
Horowitz and Burdick combined the linear structure of linearly solvable control
with sum-of-squares relaxations to obtain guaranteed polynomial over- and
under-approximations of the value function \cite{horowitz2014semidefinite}, and
Leong, Horowitz and Burdick developed linearly solvable stochastic control
Lyapunov functions whose sum-of-squares programs yield viscosity sub- and
supersolutions, pointwise value bounds and explicit suboptimality guarantees
\cite{leong2016linearly}. Occupation-measure hierarchies give lower relaxations
for nonlinear optimal control \cite{lasserre2008nonlinear}, and polynomial HJB
subsolutions have been used to derive controller-performance bounds
\cite{jones2020polynomial}. The present contribution is the source of the
bounding function: instead of approximating the physical HJB equation, we change
the control--noise geometry until it becomes exactly linearly solvable and
evaluate the surrogate by Feynman--Kac. The residual in the physical equation is
then explicit, and in the network application it names actuator directions,
prices and candidate locations.

\subsection{What the application adds}

The swing network makes the abstract distinctions unusually transparent:
additive nodal noise and direct active-power control live in the same
coordinates, inertia cancels from the matching geometry, and the deflation
defect sits entirely in frequency curvature. Under homogeneous forcing the
completion behaves as intuition about underactuation suggests: generator
channels are already matched and the least completion fills the load buses.
Under heterogeneous forcing a single quiet channel sets the temperature, and a
system with ample total authority looks severely underactuated. Pricing the same
authority in proportion to the noise removes almost all of this. ``How much
control is available?'' is therefore an incomplete question without ``in which
metric is it priced, and is that metric aligned with the noise?''

\subsection{Limitations}

Three limitations are deliberate. First, the deflated value has a rigorous
positive terminal layer and positive curvature screens, but a domain-wide
enclosure is out of reach with generic derivative bounds, which are
conservative by many orders of magnitude (\Cref{sec:app-deflation}); certifying
$\Jmin$ will need a sharper structural argument. Second, the reported crossing
probabilities are empirical; a supermartingale certificate
(\Cref{sec:app-firstexit}) has not been constructed for this event. Third, the
reinforcement experiments do not compute the price of linearisability against
an unrestricted optimal investment: the exact-PIC curves are solvable reference
values and the sparse curves are costs of a fixed controller family. None of
these affects the completed lower bound or the deployed costs. The modelling
limitations of \Cref{sec:scope} apply throughout.

\section{Conclusion and Path Forward}
\label{sec:conclusion}

The question ``how suboptimal is my controller allowed to be?'' can be answered
without computing the nonlinear stochastic optimum. An HJB subsolution bounds the
optimal cost from below, simulation of the deployed policy bounds it from above,
and their difference certifies the permissible suboptimality. Completing the
control geometry until it matches the physical noise produces such a
subsolution automatically: the completed problem is linearly solvable, its value
is a Feynman--Kac free energy, and its residual in the physical HJB equation is
exactly the energy of the fictitious control. Noise deflation is the dual
construction; it can be tighter but is conditional on a curvature sign.

The same geometry yields design rules. Because a single matching temperature is
set by the worst-matched channel, the certificate---and, we would argue, the
engineering question---is governed less by how much control authority exists
than by whether it is priced consistently with the noise. On the IEEE 118-bus
system after a severe load loss, a simple generator controller is certified
within $2.8\%$ of optimal under homogeneous forcing; under heterogeneous forcing
the gap is $39\%$ with uniform prices and $1.2\%$ after the same total authority
is priced in proportion to the noise. Exact linear solvability is an expensive
construction target but a useful oracle: it supplies shadow prices that rank
sparse reinforcement.

Several directions follow.

\paragraph{Technical improvements} A two-sided linearly solvable envelope combining completion with BOTC, extended to support mismatch, would turn
\eqref{eq:envelope} into a practical bracketing method. Matrix- or direction-dependent matching in place of one scalar temperature could remove the weakest-channel effect altogether. A structural proof of the deflation curvature condition, and a joint performance and first-exit certificate, would close the remaining gaps in \Cref{tab:status}. On the application side, measured and correlated fluctuations, multiple contingencies and explicit actuator energy, ramp and saturation limits are the next step; large computational loads, whose abrupt reductions are now an operational concern \cite{nerc2026sra}, are a natural setting in which to use measured load covariance to certify existing emergency control, reprice fast authority where possible, and only then decide where new flexibility has the highest marginal value.

\paragraph{Economic interpretation}
The prices of \Cref{sec:design} are not yet dollars. The capital side of the planning problem \eqref{eq:planning} is monetary by construction, but the value side---the cost functional, the certificate gap and $\Pi$---is measured in the weighted units of \eqref{eq:cost}. Calibrating those weights to economic costs, as production-cost and reliability studies do \cite{kirschen2018fundamentals}, would turn the shadow price of \Cref{prop:shadow} into a marginal value of fast flexibility per megawatt, comparable with capacity and performance payments in regulation markets \cite{ferc2011order755}, and $\Pi$ into the monetary cost of insisting on linearly solvable designs. It would also connect the certificates to the literature on pricing uncertainty in power systems, from chance-constrained dispatch \cite{bienstock2014chance,dvorkin2020chance} to uncertainty-aware and variance-aware prices \cite{ye2017uncertainty,mieth2020risk}. A second connection is to proportional design. Droop proportional to machine rating is standard practice \cite{kundur1994power}, and proportionality of machine parameters is what makes global analyses of synchronization performance tractable \cite{paganini2020global}. \Cref{prop:noise-proportional} asks instead for proportionality to noise variance. The two coincide when fluctuations scale with rating and differ when they do not, as for variable renewable generation and large computational loads; there the certificate suggests revisiting practice. Because repricing changes certifiability more than performance (\Cref{rem:repricing}), this is a statement about which designs can be certified tightly, not about which perform best. Placement of fast actuation in networks has also been optimized with controllability and $\mathcal H_2$ metrics \cite{pasqualetti2014controllability,summers2016submodularity,poolla2017optimal,pagnier2019optimal,tyloo2019key}; comparing those placements with the shadow-value rankings of the completed problem is a natural next test.

\paragraph{Reducing fluctuations at the source}
Installing authority is not the only lever. Buffering a variable source---storage co-located with a wind farm, or ride-through and on-site storage inside a data centre---lowers a noise variance $q_i$ instead of raising an authority $g_i$: in \Cref{fig:geometry} it is a vertical rather than a horizontal move. It is the physical counterpart of noise deflation, in the same sense that reinforcement is the physical counterpart of completion, with one advantage: once a direction's noise is physically reduced to the matched level, the deleted diffusion $\Kdef$ loses that direction, and with it that part of the curvature condition of \Cref{thm:deflation}. The scalar geometry also gives a warning. Reducing the noise at an \emph{unactuated} source lowers the completion burden and the virtual authority needed there; reducing it at the \emph{worst-matched actuated} channel lowers $\lamstar$ and can make the scalar certificate looser even as the physical system improves, so noise reduction and repricing should be planned together. The sparse experiments already test a naive version of the idea, since ranking devices by local noise variance places them at the noisiest sources. Under heterogeneous forcing in the 118-bus system the shadow ranking does better (\Cref{sec:case118}), because what matters is not the local variance of a source but its projection on the value gradient, the integrand $(a^\top\nabla J^\star)^2$ of \eqref{eq:shadow}. Where noise is largest is not necessarily where it costs the most. Variance-aware market mechanisms \cite{mieth2020risk} attach costs to sources of variability; the certificate geometry offers a dynamic, network-aware basis for doing so.

\appendix
\section{Verifying the curvature condition for the deflated value}
\label{sec:app-deflation}

By \Cref{thm:deflation}, the deflated value is a lower bound exactly when
$\rho=\tr[\Kdef_{\omega\omega}\nabla^2_{\omega\omega}\Jmin]\ge0$ holds on the whole
domain \eqref{eq:rho}. This appendix records where that condition has been
verified in the swing application, and why generic derivative bounds cannot
extend the verification to the whole domain. It is included because the partial results are
informative and because the obstruction is instructive; none of it is used for
the certificates of \Cref{sec:numerics}.

\subsection{Terminal layer}

Because $\phi=c_T\mathcal E_+/n$ and $\mathcal E\ge0$ on the cohesive cell,
$\nabla^2_{\omega\omega}\phi=(c_T/n)M$ there, and
$\rho(\cdot,T)=(c_T/n)\tr[\Kdef_{\omega\omega}M]>0$. The first two time
derivatives of $\rho$ at $T$ follow by differentiating the HJB equation, so a
second-order terminal Taylor candidate can be checked analytically on a short
terminal layer. The graph structure sharpens the compact-domain bounds this
requires. With $q=\theta-\theta^\star$ and $L_{\min}=\Inc\operatorname{diag}(K_ec_e)\Inc^\top$,
$c_e$ the minimum cosine on the admissible edge interval, the potential-energy
budget $\bar V$ gives the cycle-consistent cut-space bound
\begin{equation}
 |b_e^\top q|\le\sqrt{2\bar V\,b_e^\top L_{\min}^\dagger b_e},
 \label{eq:cut-bound}
\end{equation}
which can be intersected with edgewise Bregman bounds, and the cohesive
potential obeys the mass-scaled force estimate
\[
 \|M^{-1/2}\nabla U_{\rm rel}(\theta)\|\le\omega_0\sqrt{2V_{\rm pot}(\theta)},
 \qquad \omega_0^2=\lambda_{\max}\bigl(M^{-1/2}L_0M^{-1/2}\bigr).
\]
Using the
coupled budget $V_{\rm pot}+K_{\rm kin}\le\bar E$ rather than independent worst
cases yields a verified positive terminal layer of about $1.43$\,ms
(homogeneous) and $1.13$\,ms (activity-scaled) for case39, and about
$2.9$\,ms and $1.6$\,ms for case118 (\Cref{fig:app-terminal}). Against a
$3$\,s horizon this is a proof of concept, not a certificate.

\subsection{Curvature screens}

On a gauge-fixed compact transient set
$\{\max_e|\delta_e|\le84^\circ,\ \mathcal E/n\le0.34,\ \|\omega\|_\infty\le6.5\}$,
$34$ representative case39 states---severe uncontrolled states, mildly
controlled H3 states, spectral perturbations, several times in the first
$2.5$\,s---give positive estimates of $\rho$ in every tested Hutchinson
direction, with minimum two-standard-error lower values about
$6.45\times10^{-3}$ (homogeneous) and $4.50\times10^{-3}$ (activity-scaled)
(\Cref{fig:app-screen}). For case118, $24$ directional-curvature traces over
H0--H3 and along the H2/H3 transient are all positive, the smallest about
$9.95\times10^{-3}$. This is numerical evidence on a finite set, not an
enclosure of the continuum.

\subsection{Flow variations in a mechanical metric}

Additive noise disappears from derivatives with respect to the initial state.
Along each realization the first, second and third variations $J,H,K^{(3)}$ of
the flow of $F(\theta,\omega)$ satisfy
\begin{align}
 \dot J&=DF\,J,\qquad
 \dot H[a,b]=DF\,H[a,b]+D^2F[Ja,Jb],\label{eq:var12}\\ \nonumber 
 \dot K^{(3)}[a,b,c]&=DF\,K^{(3)}[a,b,c]+D^2F[H[a,b],Jc]+D^2F[H[a,c],Jb]\\ &
 +D^2F[H[b,c],Ja]+D^3F[Ja,Jb,Jc],\nonumber
\end{align}
with, for $z_j=(\eta_j,v_j)$ and $\delta=\Inc^\top\theta$,
\begin{equation}
 DF=\begin{bmatrix}0&I\\-M^{-1}L(\theta)&-M^{-1}C\end{bmatrix},\quad
 D^2F[z_1,z_2]=\begin{bmatrix}0\\M^{-1}\Inc\,
 \bigl[K_e\sin\delta_e(\Inc^\top\eta_1)_e(\Inc^\top\eta_2)_e\bigr]_e\end{bmatrix},
 \label{eq:derivs}
\end{equation}
and $D^3F[z_1,z_2,z_3]$ equal to the same expression with $\sin\delta_e$ replaced
by $\cos\delta_e$ and the product taken over three factors. Only angle
components enter $D^2F$ and $D^3F$. On a cohesive cell centred at $\theta_c$, the
frozen linearization is dissipative in the mechanical metric
$\|(\eta,v)\|_c^2=\eta^\top L(\theta_c)\eta+v^\top Mv$ (after fixing the uniform
phase), so positive growth is driven only by the variation of $L(\theta)$ within
the cell. If $\mu_c,C_{2,c},C_{3,c}$ bound the first three local derivative
operators, the scalar majorants $\dot j=\mu_cj$, $\dot h=\mu_ch+C_{2,c}j^2$,
$\dot k=\mu_ck+3C_{2,c}hj+C_{3,c}j^3$ propagate the sensitivities while the
trajectory stays in the cell, with metric-equivalence constants at cell
transitions. For a representative case39 H3 transient over $2.5$\,s, a
$0.10^\circ$ cell gives a first-variation majorant near $20.5$, against about
$4.35\times10^{74}$ for a single $84^\circ$ sector bound (\Cref{fig:app-cellwise}).

\subsection{Why a domain-wide enclosure is out of reach}

A covering argument needs Lipschitz control of $\nabla^2\Jmin$, hence of the
third variation, on every cell of a finite cover of the $78$-dimensional (case39)
state--time set. An adaptive-cell prototype for case118 shows that the generic
third-derivative closure is unusable: even under optimistic local constants the
certified cell radii collapse to about $9\times10^{-12}$ degrees for homogeneous
H3 and $2\times10^{-14}$ degrees for activity-scaled H3. Finer cells will not fix
this. A certificate for $\Jmin$ will need a qualitatively sharper argument---a
semigroup or comparison principle for $\rho$, a monotonicity structure, or a
covariance-sensitive bound using the Gibbs representation of \Cref{box:physics}.

\begin{figure}[t]
  \centering
  \includegraphics[width=0.46\textwidth]{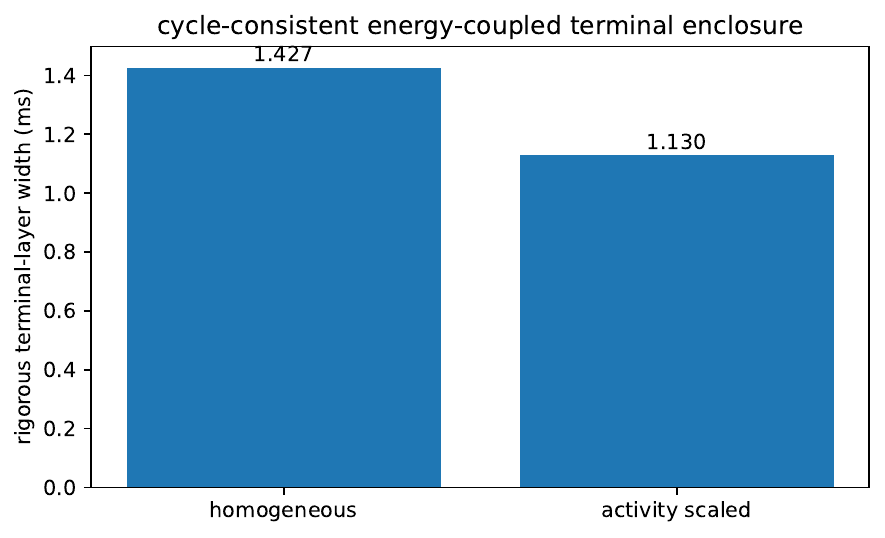}\hfill
  \includegraphics[width=0.50\textwidth]{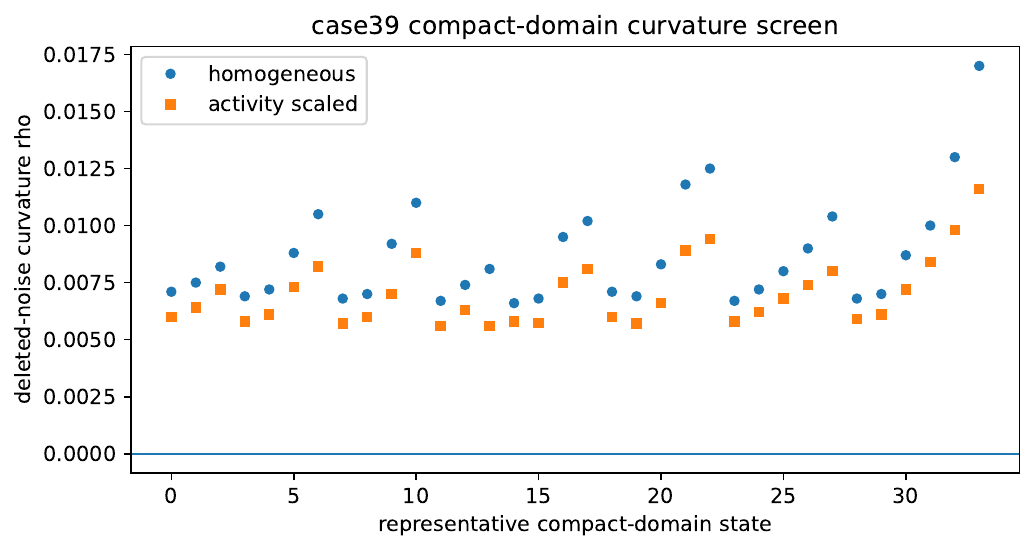}
  \caption{Left: width of the rigorous positive terminal curvature layer for the
  two case39 noise models, from the cycle-consistent energy-coupled enclosure.
  Right: estimated deleted-noise curvature $\rho$ on the representative case39
  compact-domain states; positive throughout, but a finite-set screen, not an
  enclosure.}
  \label{fig:app-terminal}
  \label{fig:app-screen}
\end{figure}

\begin{figure}[t]
  \centering
  \includegraphics[width=0.55\textwidth]{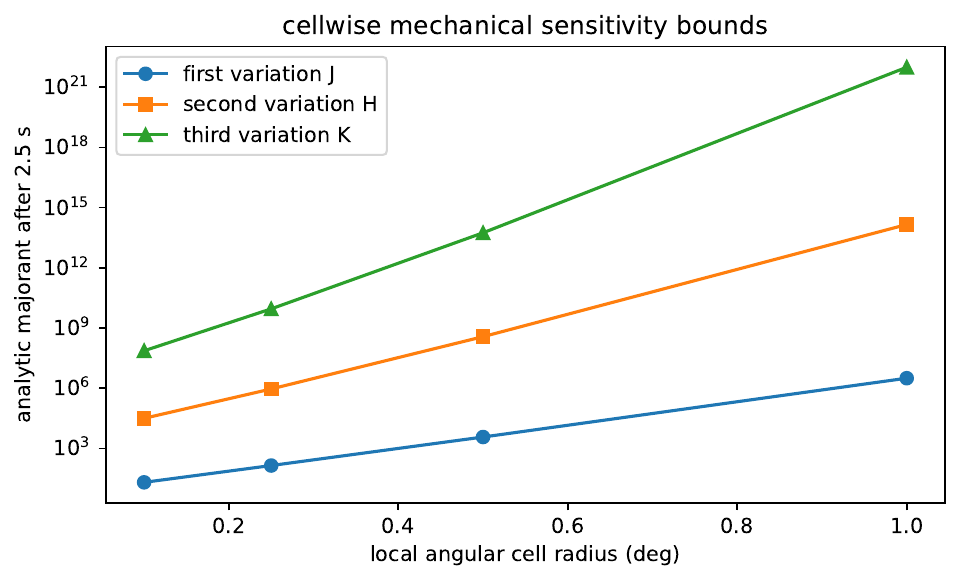}
  \caption{First-, second- and third-variation majorants after $2.5$\,s against
  the local angular cell radius, for a representative case39 H3 transient.
  Preserving the local mechanical metric is essential; a single global sector
  estimate destroys the dissipative structure.}
  \label{fig:app-cellwise}
\end{figure}

\section{First-exit probability as a separate certificate}
\label{sec:app-firstexit}

Let $\Omega_{\rm coh}$ be the cohesive domain and $\tau$ its first-exit time.
For a deployed control $u$, a nonnegative $Z\in C^{2,1}$ with
\begin{equation}
 \partial_tZ+(f+Bu)\cdot\nabla Z+\tr(D\nabla^2Z)\le0\ \text{ in }\Omega_{\rm coh},
 \qquad Z\ge1\ \text{ on }\partial\Omega_{\rm coh},
 \label{eq:supermartingale}
\end{equation}
makes $Z(X_{s\wedge\tau},s\wedge\tau)$ a supermartingale, and optional stopping
gives $\Prob^u_{x,t}\{\tau\le T\}\le Z(x,t)$. This is the stopped-domain
counterpart of \Cref{sec:certificates} with the inequality reversed. It is the
appropriate framework for a certified synchrony-loss probability and should not
be confused with the smooth performance value used in the main text; we have
not constructed such a $Z$ for the event studied here, and the crossing
probabilities of \Cref{sec:numerics} are empirical.

\section{Numerical protocol and reproducibility}
\label{sec:app-repro}

\paragraph{Code and data}
All code, frozen inputs, cached result tables and notebooks are available in a
public repository \url{https://github.com/mchertkov/SuboptimalController}.
The case39 study and the case118 study are separate packages. Each regenerates
every figure and table of the paper from cached results in seconds, and
documents the heavy recomputation (Monte Carlo campaigns of several hours) in
execution order. A third notebook checks the theorems of
\Cref{sec:relaxations,sec:computing} on a linear--quadratic example with known
closed-form value.

\paragraph{Seeds and sampling}
Tuning and evaluation seeds are frozen and disjoint; paired policy comparisons
use common random numbers. Case118 deployed costs use $192$ evaluation paths per
policy and noise level, except where an independent audit is stated
($3072$ paths). Completed values are importance-sampled with guides based on the
deployed or completed closed loop, and we report replicate standard errors and
effective sample sizes.

\paragraph{Time step}
Load buses carry small synthetic inertia, which makes the reduced model stiff.
In case39, $\Delta t\ge10^{-2}$\,s produces spurious angle excursions and a
biased cost (Figure~\ref{fig:sm-convergence} in \cref{app:figures}); production runs use
$\Delta t\le5\times10^{-3}$\,s, with finer steps for the most sensitive checks,
and case118 uses $\Delta t=2.5\times10^{-3}$\,s.

\paragraph{Test systems}
Network data are the standard \textsc{matpower} case files
\cite{zimmerman2011matpower}, reduced to the lossless fixed-voltage swing model
of \Cref{sec:swing}; the deterministic post-event equilibria reproduce those of
\cite{chertkov2026cyclespace}.

\section{Additional figures}
\label{app:figures}
These twelve figures support Section~\ref{sec:numerics}, where
each is cited. Each is followed by two short paragraphs, labelled by the figure number, stating
what it shows and the conclusion drawn from it. Notation follows the main text: $J$ is the cost of a
deployed controller, $\Jp$ the completed value (an exact lower bound,
Theorem~\ref{thm:completion}), $\Jmin$ the deflated value (a lower bound only
under the curvature condition of Theorem~\ref{thm:deflation}), $\lamstar$ the
matching temperature and $\Ip$ the completion burden. All figures are generated from cached results by the scripts in the accompanying repository (\cref{sec:app-repro}).
Figures~\ref{fig:sm-noise}--\ref{fig:sm-sparse39} concern the 39-bus system,
Figures~\ref{fig:sm-radial} and~\ref{fig:sm-convergence} the formulation and the
numerics, and Figure~\ref{fig:sm-deployed118} the 118-bus system.

\subsection{The 39-bus system: deployed controllers}
\label{sec:sm-39-controllers}

\begin{figure}[htbp]
  \centering
  \includegraphics[width=0.85\textwidth]{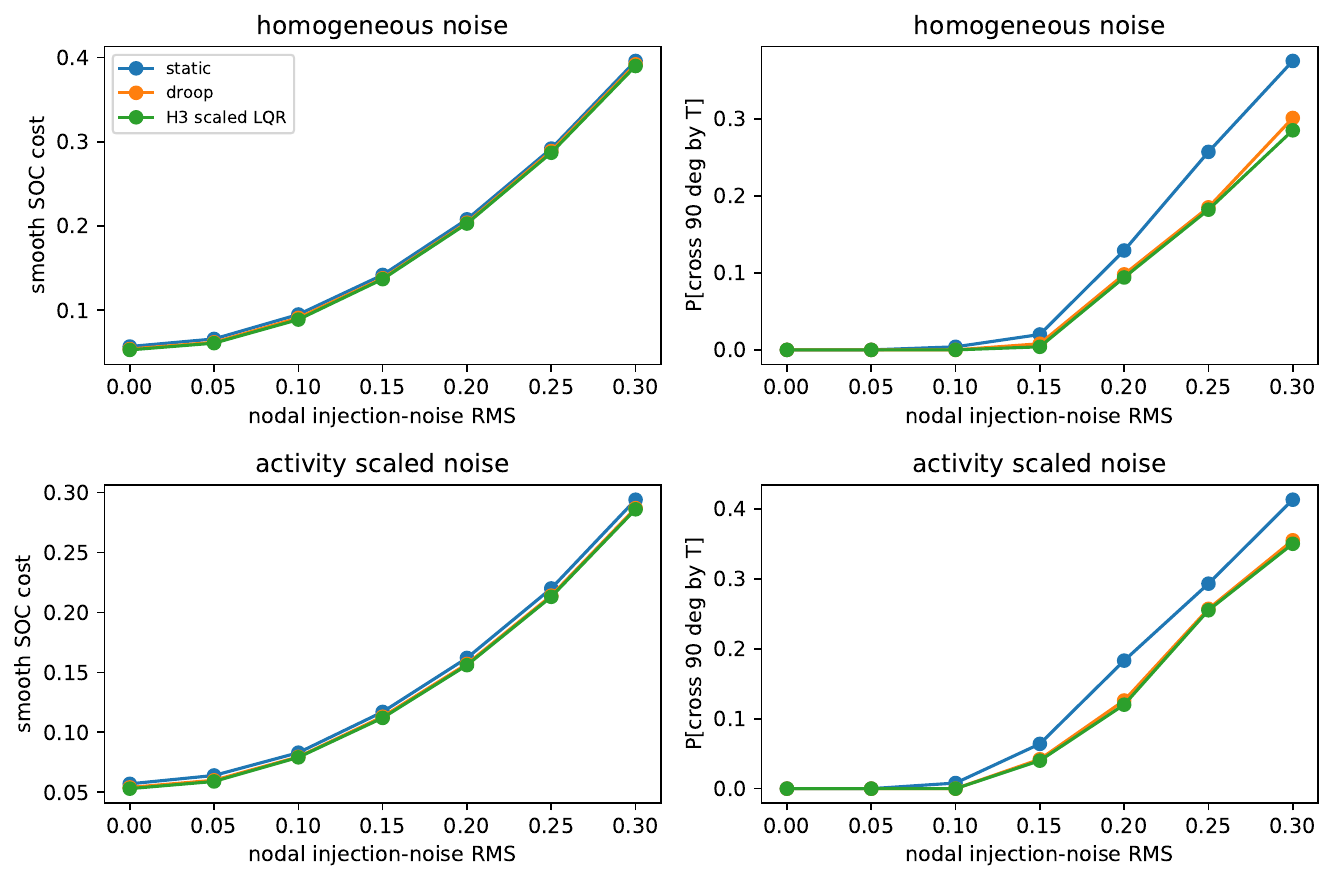}
  \caption{Case39 noise-amplitude sweeps for homogeneous (top) and activity-scaled (bottom) nodal noise. Left: smooth cost $J$. Right: empirical probability of crossing $90^\circ$ within the horizon $T=3$\,s.}
  \label{fig:sm-noise}
\end{figure}

\paragraph{Figure~\ref{fig:sm-noise}: what is shown}
The three deployed policies---static balancing, generator droop and H3 scaled
LQR---are evaluated on common random paths as the noise amplitude
$\sigma_{\rm rms}$ grows from $0$ to $0.3$. The smooth cost rises roughly
sevenfold for homogeneous noise and fivefold for activity-scaled noise, and the
three policies are almost indistinguishable on this scale. The crossing
probability is negligible up to $\sigma_{\rm rms}\approx0.10$ and then rises
steeply: under homogeneous noise it reaches $0.129$, $0.098$ and $0.094$ at
$\sigma_{\rm rms}=0.2$ and $0.375$, $0.301$ and $0.285$ at $0.3$; under
activity-scaled noise the rise starts earlier.

\paragraph{Figure~\ref{fig:sm-noise}: what we conclude}
The weak controllers barely change the smooth objective but remove about a
quarter of the boundary crossings at high noise, so performance and safety must
be reported separately. At equal root-mean-square amplitude the two spatial
profiles give different safety statistics: where the noise sits matters, not
only how large it is. This is why the main text keeps both noise models
throughout.

\begin{figure}[htbp]
  \centering
  \includegraphics[width=0.6\textwidth]{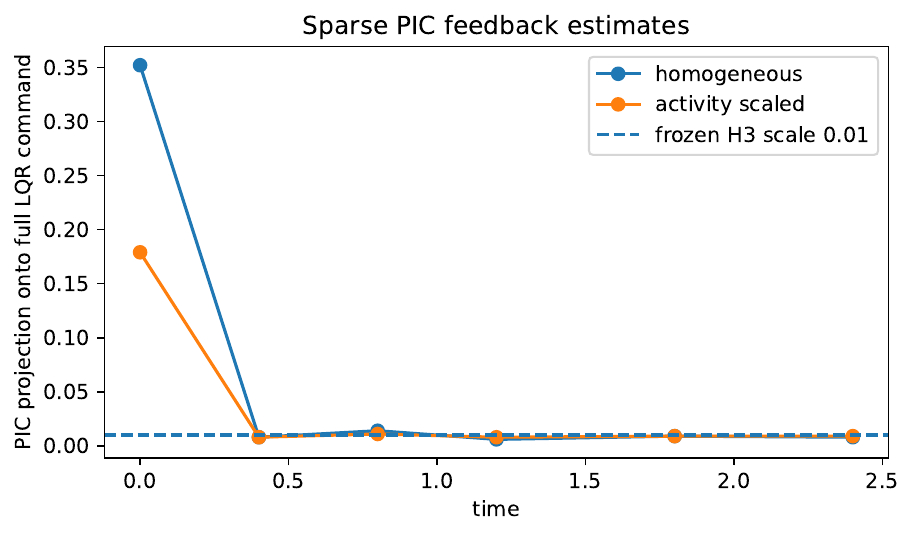}
  \caption{Case39 projection of sparse path-integral feedback estimates onto the full H3 LQR command, as a function of time. The dashed line is the H3 LQR scale $0.01$ selected by the upper-cost scan.}
  \label{fig:sm-projection}
\end{figure}

\paragraph{Figure~\ref{fig:sm-projection}: what is shown}
The path-integral estimate of the optimal feedback is projected onto the
direction of the H3 LQR command at several times along the horizon. At $t=0$ the
projection is large ($0.35$ homogeneous, $0.18$ activity-scaled); from
$t\approx0.4$\,s on it settles close to $0.01$ for both noise models.

\paragraph{Figure~\ref{fig:sm-projection}: what we conclude}
Away from the estimator singularity at $t=0$, where the feedback is estimated
from paths that have not yet separated, an independent path-integral
calculation supports the same feedback scale as the direct cost scan. The
deployed H3 controller is therefore not an arbitrary choice: it sits where the
path-integral feedback points.

\subsection{The 39-bus system: the two relaxations across the hierarchy}
\label{sec:sm-39-hierarchy}

\begin{figure}[htbp]
  \centering
  \includegraphics[width=0.85\textwidth]{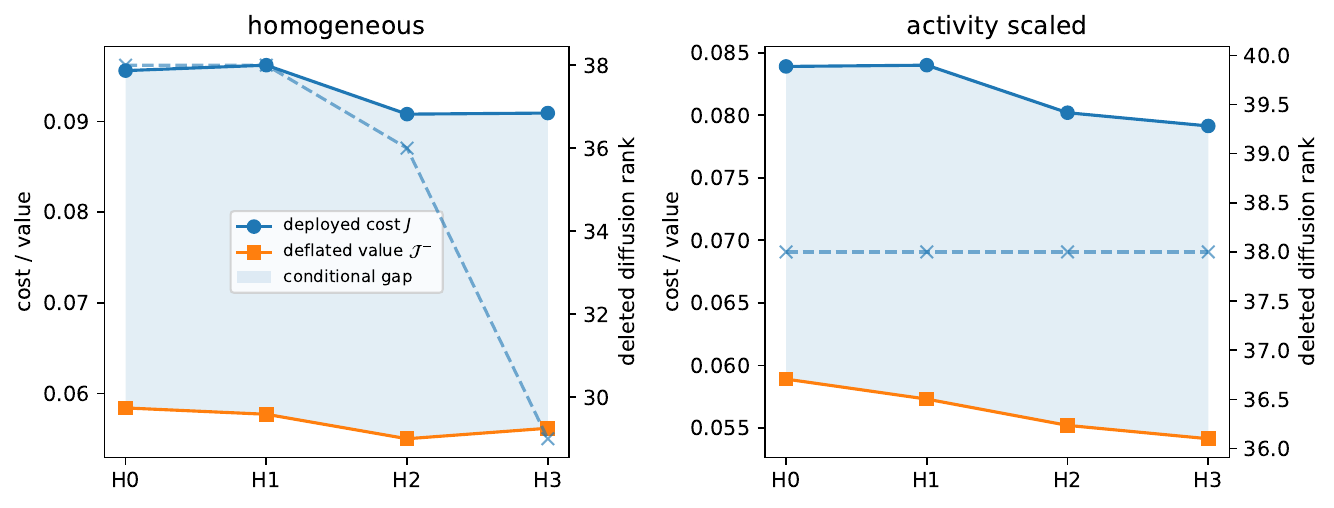}
  \caption{Case39 H0--H3 comparison at $\sigma_{\rm rms}=0.1$: deployed cost $J$ (upper curve) and deflated value $\Jmin$ (lower curve), with the rank of the deleted stochastic subspace on the secondary axis (dashed). The shading is the gap $J-\Jmin$, which is a certified gap only under the curvature condition of Theorem~\ref{thm:deflation}.}
  \label{fig:sm-h0h3}
\end{figure}

\paragraph{Figure~\ref{fig:sm-h0h3}: what is shown}
For homogeneous noise the deployed cost falls from about $0.096$ (H0, H1) to
$0.091$ (H2, H3) while $\Jmin$ stays near $0.055$--$0.058$, and the rank of the
deleted subspace falls from $38$ to $36$ and $29$ as actuator directions are
added. For activity-scaled noise the deployed cost falls from $0.084$ to $0.079$,
$\Jmin$ from $0.059$ to $0.054$, and the deleted rank stays at $38$ for every
architecture.

\paragraph{Figure~\ref{fig:sm-h0h3}: what we conclude}
Under homogeneous noise, adding actuators removes stochastic directions from the
deleted subspace one for one, as intuition about underactuation suggests. Under
heterogeneous noise deflation deletes---at least partly---$38$ of the $39$
stochastic directions whatever the architecture: one channel sets $\lamstar$,
and every other channel keeps an unmatched surplus that the deflation must
discard. Together with the fact that its lower-bound property is conditional,
this is why the main text bases its certificates on completion rather than
deflation.

\begin{figure}[htbp]
  \centering
  \includegraphics[width=0.85\textwidth]{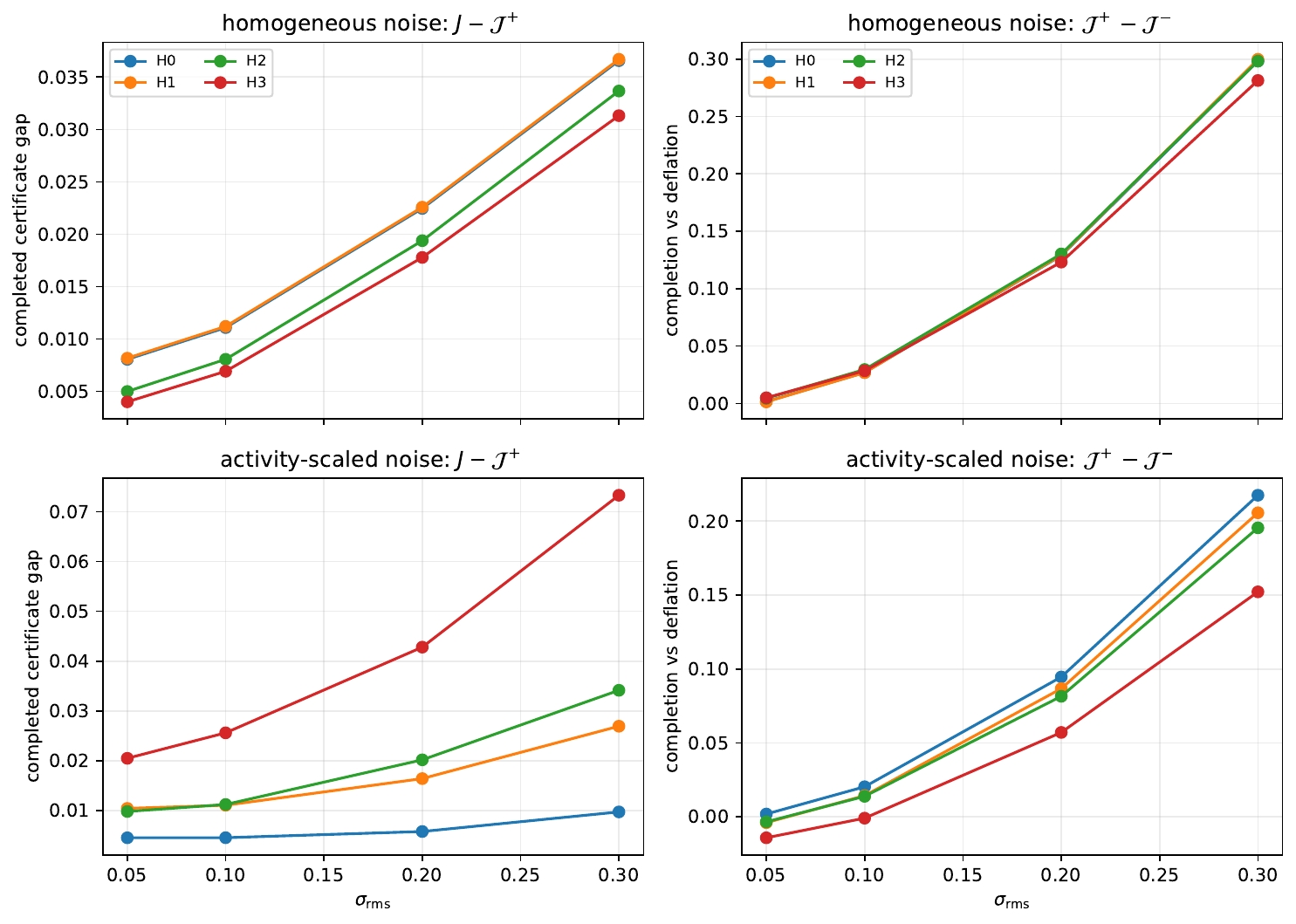}
  \caption{Case39 completion across the noise amplitude, for H0--H3. Left: completed certificate gap $J-\Jp$. Right: difference $\Jp-\Jmin$ between the two relaxations. Top: homogeneous noise; bottom: activity-scaled noise.}
  \label{fig:sm-inflation-noise}
\end{figure}

\paragraph{Figure~\ref{fig:sm-inflation-noise}: what is shown}
The completion geometry---$G_+=Q/\lamstar$ and the diagnostics $\Ip,\Dp$---is
invariant under a common rescaling $\Gamma\mapsto s\Gamma$, but the values are
not. Under homogeneous noise the certified gap of H3 grows from $0.0040$ to
$0.0069$, $0.0178$ and $0.0313$ at $\sigma_{\rm rms}=0.05,0.1,0.2,0.3$, with the
architecture ordering preserved (H0 and H1 loosest, H3 tightest). Under
activity-scaled noise H3 becomes the \emph{loosest} architecture, with gaps
$0.0205$, $0.0256$, $0.0428$ and $0.0733$, while H0 is the tightest. The
difference $\Jp-\Jmin$ grows with the noise in both cases; at the weakest
heterogeneous noise it is negative for H1--H3, most strongly for H3.

\paragraph{Figure~\ref{fig:sm-inflation-noise}: what we conclude}
Stronger forcing makes every certificate looser, but the ranking of
architectures is decided by the geometry, not by the amplitude. Under
heterogeneous noise the weakest-channel effect makes the architecture with the
most physical control the one with the most optimistic scalar completion. The
sign change of $\Jp-\Jmin$ is a concrete instance of Remark~\ref{rem:not-ordered}: neither relaxation dominates the other in general.

\subsection{The 39-bus system: planning}
\label{sec:sm-39-planning}

\begin{figure}[htbp]
  \centering
  \includegraphics[width=0.85\textwidth]{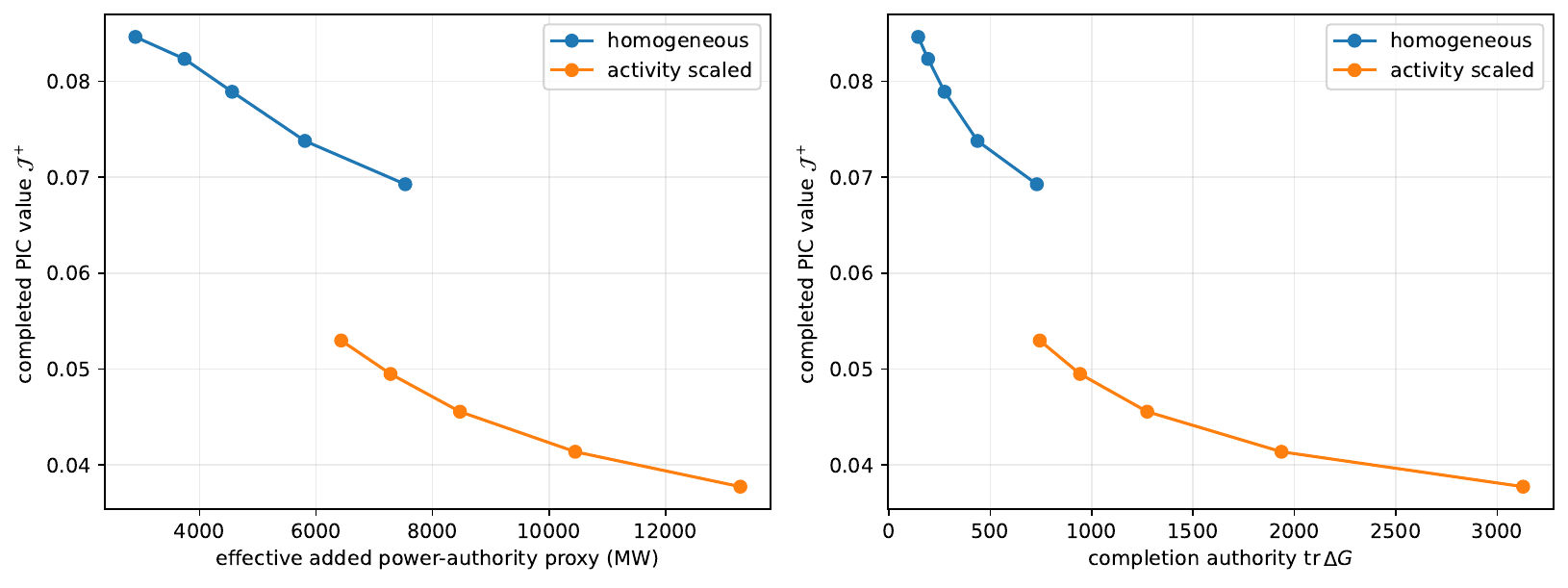}
  \caption{Case39 linearly solvable reinforcement frontier from H3: completed value $\Jp_\lambda$ against the effective power-authority proxy (left) and against the completion trace (right), as $\lambda$ decreases below $\lamstar$.}
  \label{fig:sm-frontier39}
\end{figure}

\paragraph{Figure~\ref{fig:sm-frontier39}: what is shown}
Each curve starts at the least completion $\lambda=\lamstar$ and moves to
$\lambda/\lamstar=0.8,0.6,0.4,0.25$. Under homogeneous noise $\Jp$ falls from
$0.0846$ to $0.0693$ while the authority proxy grows from about $2.9$ to
$7.5$\,GW and the trace from $145$ to about $730$. The activity-scaled curve
starts lower, near $0.053$, needs roughly twice the proxy authority at every
point, and falls to about $0.038$.

\paragraph{Figure~\ref{fig:sm-frontier39}: what we conclude}
The investment--performance frontier of equation~\eqref{eq:frontier} is monotone, as Proposition~\ref{prop:least-inflation} guarantees, with
diminishing returns. Under heterogeneous noise the same operational gain costs
much more authority, because the small matching temperature inflates every
direction. Exact completion is useful as a reference frontier; as a
construction target it is expensive.

\begin{figure}[htbp]
  \centering
  \includegraphics[width=0.8\textwidth]{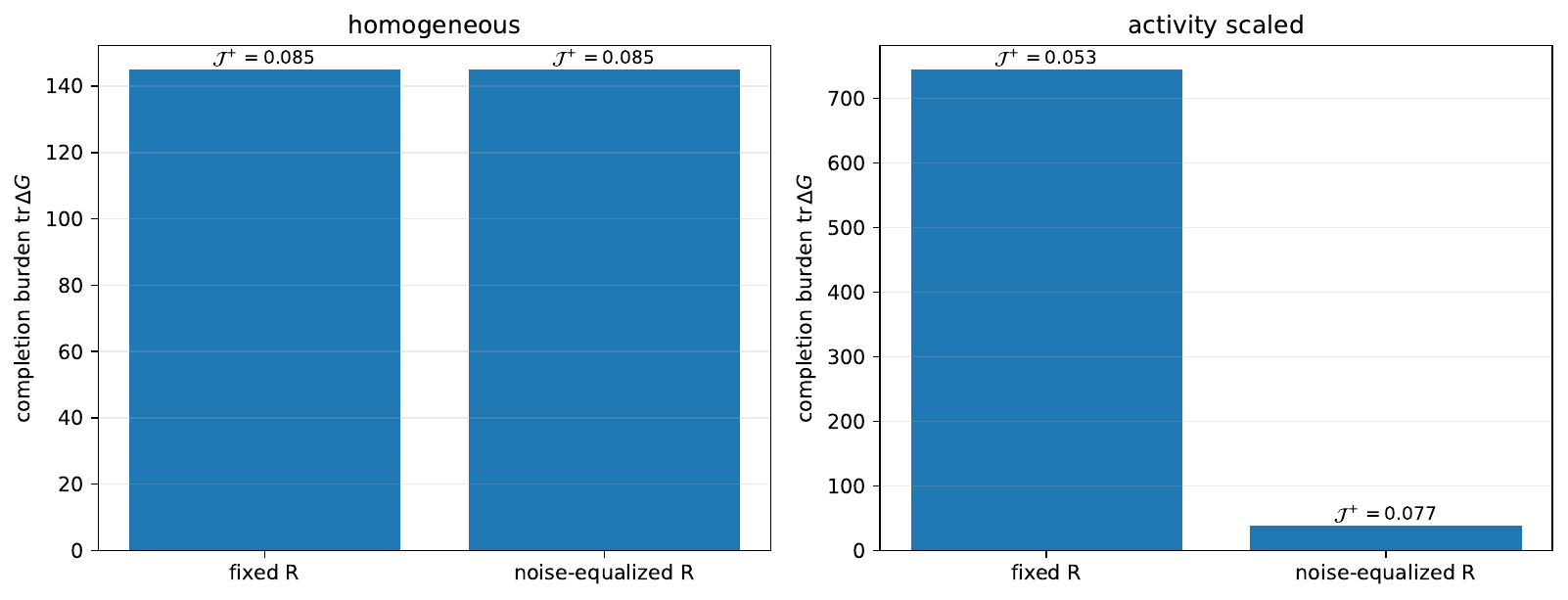}
  \caption{Case39 pricing of the existing generator authority at fixed total authority $\sum_gg_g=50$: completion burden $\tr\Delta G$ (bars) and completed value $\Jp$ (annotations), with the original uniform metric and with the noise-equalized metric of Proposition~\ref{prop:noise-proportional}.}
  \label{fig:sm-metric}
\end{figure}

\paragraph{Figure~\ref{fig:sm-metric}: what is shown}
Under homogeneous noise the two metrics coincide: the burden is $145$ and
$\Jp=0.085$ in both. Under activity-scaled noise the uniform metric has burden
about $744$ and $\Jp=0.053$, and the noise-equalized metric has burden $38$ and
$\Jp=0.077$, without any added authority.

\paragraph{Figure~\ref{fig:sm-metric}: what we conclude}
With the deployed cost near $0.080$ in both cases, repricing moves the certified
gap from about a third of the cost to a few per cent. Under homogeneous noise
there is nothing to gain, because uniform prices are already noise-proportional.
As Remark~\ref{rem:repricing} explains, repricing changes the
control problem, so this is a design statement: with noise-proportional prices,
the simple controller is certified tightly.

\begin{figure}[htbp]
  \centering
  \includegraphics[width=0.8\textwidth]{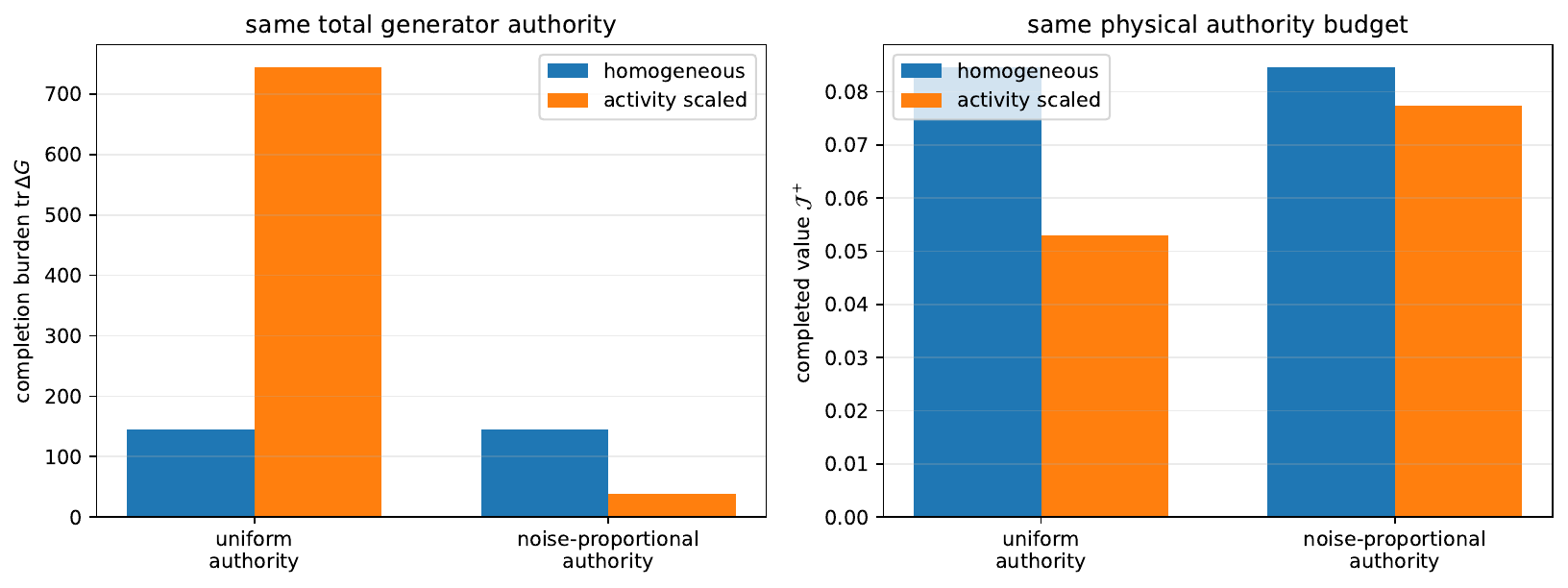}
  \caption{Case39 generator authority allocations at the same total authority, uniform and noise-proportional, for both noise shapes. Left: completion burden. Right: completed value $\Jp$.}
  \label{fig:sm-allocation}
\end{figure}

\paragraph{Figure~\ref{fig:sm-allocation}: what is shown}
The same two allocations as in Figure~\ref{fig:sm-metric}, arranged by
allocation rather than by noise shape so that the two scenarios can be compared
directly. The noise-proportional allocation leaves the homogeneous case
unchanged (burden $145$, $\Jp\approx0.085$) and changes the activity-scaled case
from burden about $745$ and $\Jp\approx0.053$ to burden $38$ and
$\Jp\approx0.077$.

\paragraph{Figure~\ref{fig:sm-allocation}: what we conclude}
The allocation of Proposition~\ref{prop:noise-proportional} is never worse than
uniform and is decisive when the noise is heterogeneous. Because it is tuned to
one noise shape, it raises the question answered by the next figure: what to do
when the noise shape is not known in advance.

\begin{figure}[htbp]
  \centering
  \includegraphics[width=0.85\textwidth]{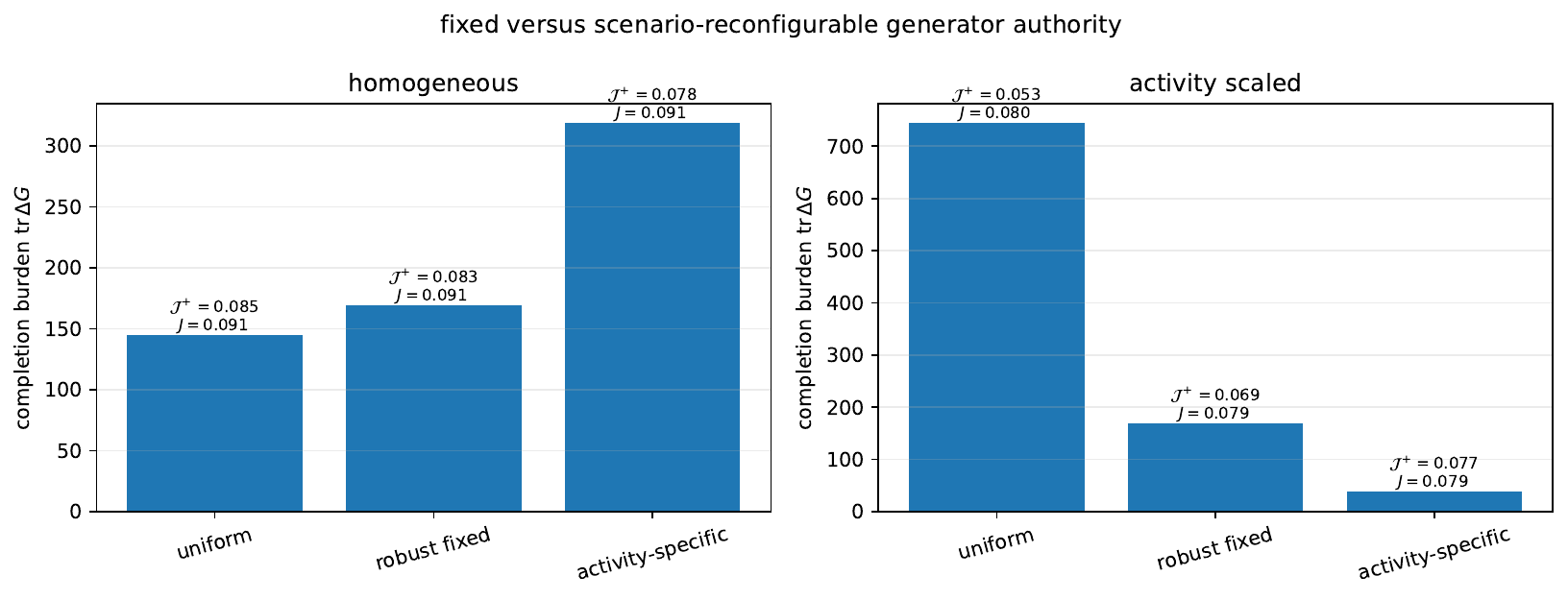}
  \caption{Case39 scenario flexibility: completion burden (bars) and the values $\Jp$ and $J$ (annotations) for uniform, robust fixed and activity-specific allocations of the same total authority, under homogeneous (left) and activity-scaled (right) noise.}
  \label{fig:sm-scenario}
\end{figure}

\paragraph{Figure~\ref{fig:sm-scenario}: what is shown}
Under homogeneous noise the burden is $145$ (uniform), $169$ (robust) and about
$320$ (activity-specific), with $\Jp=0.085$, $0.083$ and $0.078$. Under
activity-scaled noise it is about $745$, $169$ and $38$, with $\Jp=0.053$, $0.069$
and $0.077$. Within each scenario the deployed cost is nearly the same across
allocations ($0.091$, and $0.079$--$0.080$).

\paragraph{Figure~\ref{fig:sm-scenario}: what we conclude}
Each scenario-specific allocation is best for its own noise shape and noticeably
worse for the other, as Proposition~\ref{prop:multiscenario}
predicts: one fixed architecture cannot be exactly matched to two covariance
shapes. The robust allocation of Proposition~\ref{prop:robust} equalizes the
burden across the two scenarios at $169$ and is never badly mismatched. If prices
can be reset between forecast scenarios, the scenario-specific allocation gives
the tightest certificate; if they cannot, the robust one is the natural
compromise.

\begin{figure}[htbp]
  \centering
  \includegraphics[width=0.85\textwidth]{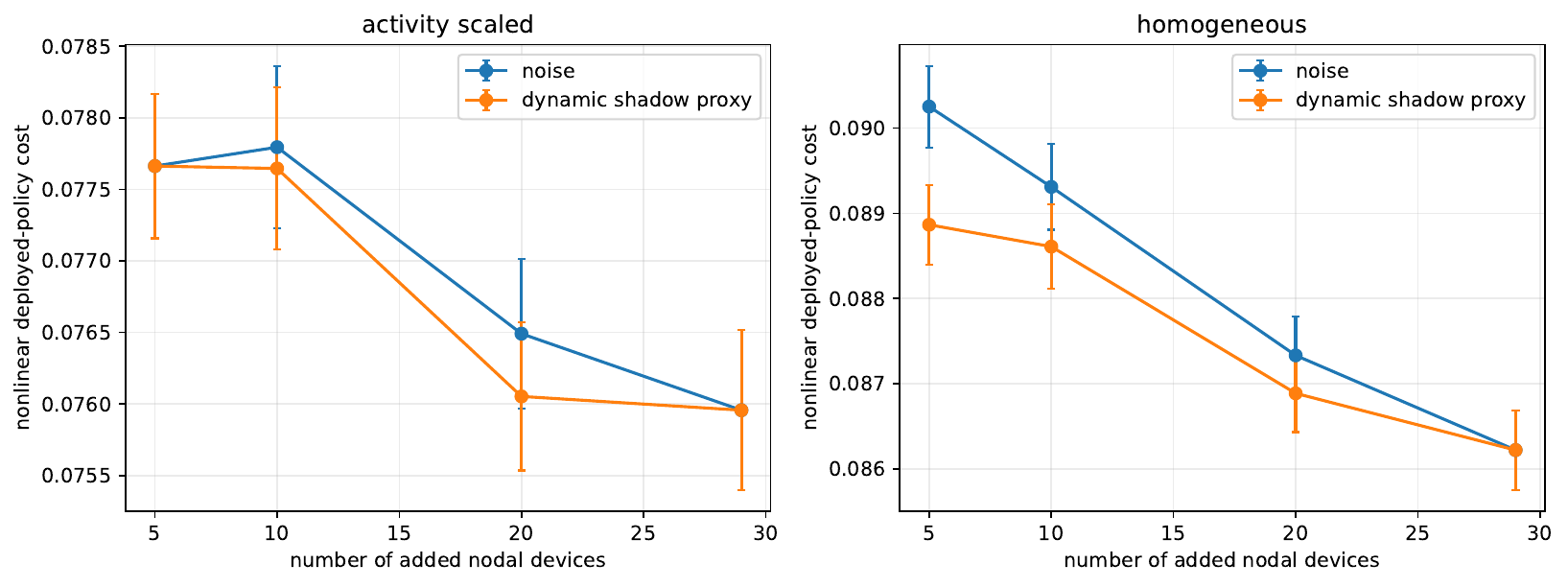}
  \caption{Case39 sparse load-side reinforcement: deployed nonlinear cost against the number of added fast devices at load buses, placed by local noise variance and by a dynamics-aware shadow proxy, for activity-scaled (left) and homogeneous (right) noise. Error bars are Monte Carlo standard errors.}
  \label{fig:sm-sparse39}
\end{figure}

\paragraph{Figure~\ref{fig:sm-sparse39}: what is shown}
Under homogeneous noise, five devices placed by the shadow proxy lower the
deployed cost to $0.08887\pm0.00047$, against $0.09025\pm0.00048$ for five placed
by noise and about $0.091$ with generators only; the advantage shrinks as devices
are added and vanishes at full completion ($29$ devices, $0.0862$). Under
activity-scaled noise both rankings lower the cost from about $0.0777$ to
$0.0760$ and stay within about one standard error of each other at every
budget.

\paragraph{Figure~\ref{fig:sm-sparse39}: what we conclude}
Shadow prices from the completed problem (Proposition~\ref{prop:shadow}) can
direct a small device budget to where it matters, and they do so here under
homogeneous noise. The pattern is system dependent: in the 118-bus system
(Section~\ref{sec:case118}) the shadow ranking is clearly better under
heterogeneous noise, while under homogeneous noise, where the noise ranking
reduces to an arbitrary ordering, that arbitrary ordering does slightly better
at intermediate budgets. Dynamics-aware placement is therefore a useful
heuristic informed by the certificate geometry, not a guarantee.

\subsection{Formulation and numerics}
\label{sec:sm-numerics}

\begin{figure}[htbp]
  \centering
  \includegraphics[width=0.85\textwidth]{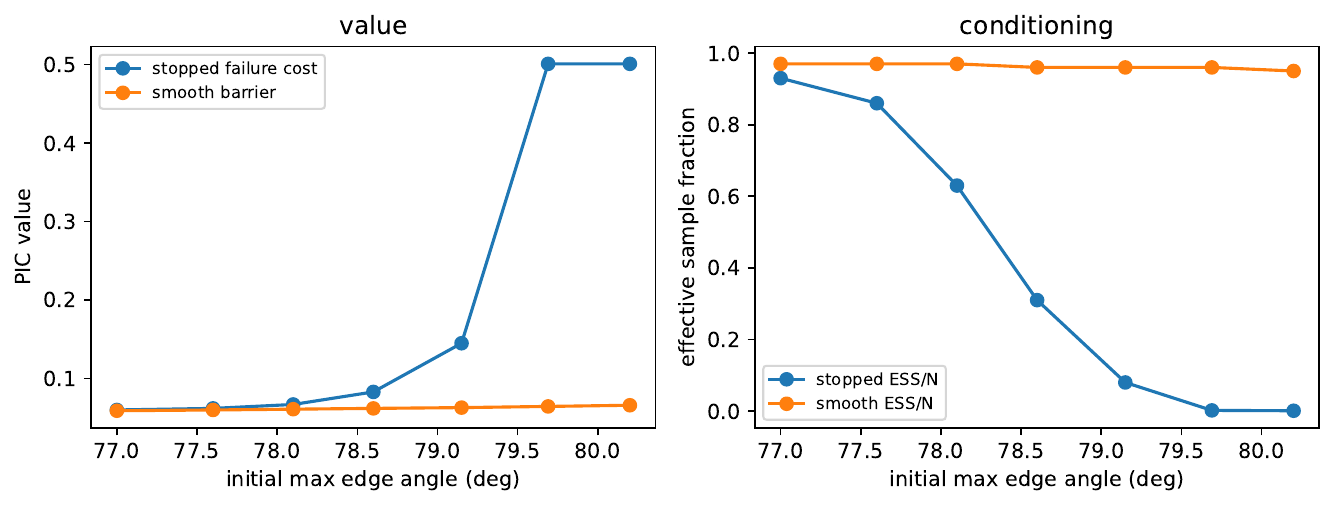}
  \caption{Stopped failure-cost and smooth-barrier Feynman--Kac formulations on a radial family of case39 initial states approaching the cohesion boundary. Left: path-integral value. Right: effective sample fraction $\ESS/N$.}
  \label{fig:sm-radial}
\end{figure}

\paragraph{Figure~\ref{fig:sm-radial}: what is shown}
Initial states are placed along a ray whose maximum line angle rises from
$77^\circ$ to $80^\circ$. With an absorbing failure cost the path-integral value
rises steeply and saturates beyond about $79.5^\circ$, and the effective sample
fraction collapses from $0.93$ to essentially zero. With the smooth barrier of
equation~\eqref{eq:power-cost} the value changes only slightly
and the effective sample fraction stays near $0.96$ throughout.

\paragraph{Figure~\ref{fig:sm-radial}: what we conclude}
An absorbing failure cost turns the Feynman--Kac estimate into a rare-survival
problem within a few degrees of the boundary, where the value is dominated by
the few paths that happen to survive. The smooth objective avoids this and keeps
the estimator usable in exactly the stressed regime of interest. This is the
reason for separating smooth performance, which is certified, from synchrony
safety, which would need the separate construction of
Appendix~\ref{sec:app-firstexit}.

\begin{figure}[htbp]
  \centering
  \includegraphics[width=0.8\textwidth]{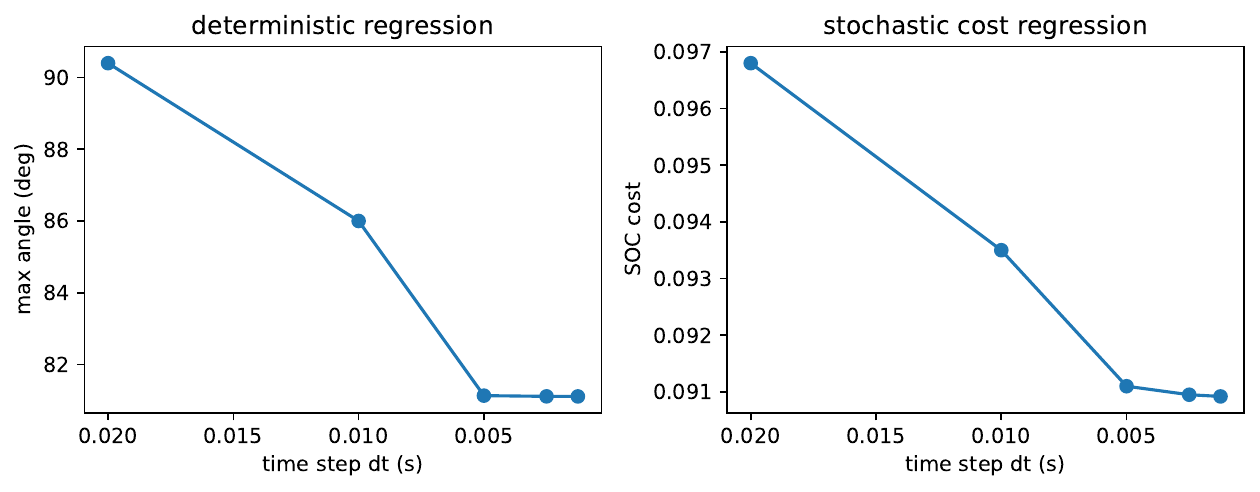}
  \caption{Case39 time-step regression. Left: deterministic maximum line angle. Right: nominal stochastic cost. The time step decreases to the right.}
  \label{fig:sm-convergence}
\end{figure}

\paragraph{Figure~\ref{fig:sm-convergence}: what is shown}
At $\Delta t=0.02$\,s the deterministic transient reaches about $90.4^\circ$ and
the stochastic cost is about $0.0968$; at $0.01$\,s these are $86.0^\circ$ and
$0.0935$; for $\Delta t\le5\times10^{-3}$\,s they settle at $81.1^\circ$ and about
$0.0909$--$0.0911$.

\paragraph{Figure~\ref{fig:sm-convergence}: what we conclude}
The small synthetic inertia of the load buses makes the reduced model stiff.
Coarse steps are not merely inaccurate but qualitatively wrong: at
$\Delta t=0.02$\,s they report a crossing of the cohesion boundary that does not
occur, and they overstate the cost by about $6\%$. The production step
$\Delta t\le5\times10^{-3}$\,s is converged to about $0.2\%$ in the cost.

\subsection{The 118-bus system: deployed controllers}
\label{sec:sm-118}

\begin{figure}[htbp]
  \centering
  \includegraphics[width=0.85\textwidth]{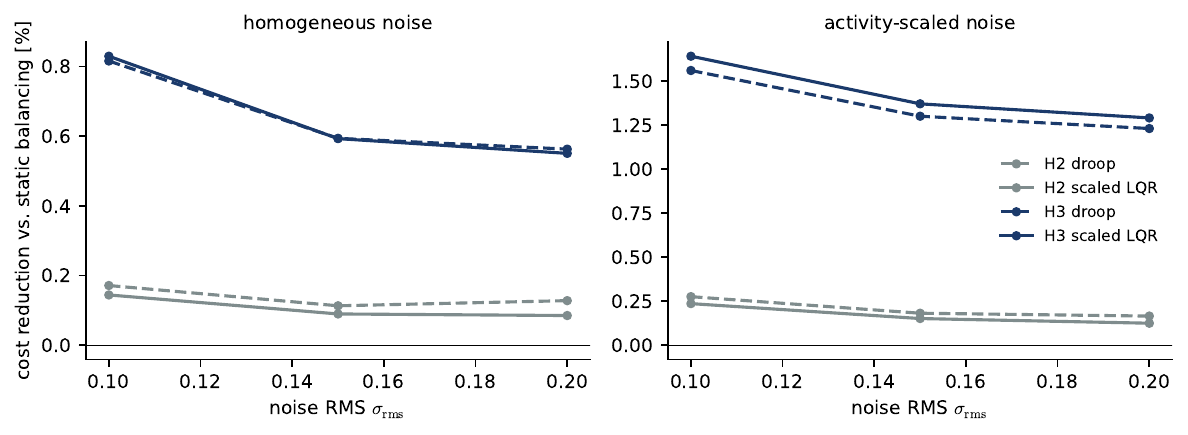}
  \caption{Case118 smooth-cost reduction of droop and scaled-LQR controllers relative to static balancing, for the H2 and H3 architectures, under homogeneous (left) and activity-scaled (right) noise, on common random paths.}
  \label{fig:sm-deployed118}
\end{figure}

\paragraph{Figure~\ref{fig:sm-deployed118}: what is shown}
Under homogeneous noise H3 lowers the cost relative to static balancing by
$0.83\%$, $0.59\%$ and $0.55\%$ at $\sigma_{\rm rms}=0.10,0.15,0.20$ with scaled
LQR, and almost identically with droop; H2 achieves only $0.09$--$0.17\%$. Under
activity-scaled noise H3 achieves $1.64\%$, $1.37\%$ and $1.29\%$ with scaled LQR
and slightly less with droop, and H2 again only $0.12$--$0.28\%$.

\paragraph{Figure~\ref{fig:sm-deployed118}: what we conclude}
In this deliberately weak-feedback regime, actuator rank matters far more than
the control law: the difference between H2 and H3 dwarfs the difference between
droop and scaled LQR. The improvements are also small in absolute terms, which
is precisely why a certificate is needed. It is the certified gap of
Table~\ref{tab:headline}, not the improvement over static
balancing, that tells an operator whether a better controller is worth
designing.

\section*{Acknowledgments}
The author thanks Bert Kappen for discussions in August 2026 that stimulated the
development and testing of the theory.

\paragraph{Disclaimer on the use of language models.}
A large language model was used to proof-read the manuscript and to help design,
implement and check the numerical experiments. All results were reviewed by the
author, who is solely responsible for the content, the claims and any errors.

\bibliographystyle{plainnat}
\bibliography{refs}
\end{document}